\documentclass[a4paper, 12pt]{article}
\usepackage[left=2.4cm,right=2.4cm,top=3cm]{geometry}

\usepackage{amsmath}
\usepackage{algorithm}
\usepackage{algpseudocode}
\usepackage{amsfonts}
\usepackage{amssymb}
\usepackage{amsthm} 
\usepackage{multicol}
\usepackage{graphicx}
\usepackage{enumerate}
\usepackage{booktabs}
\usepackage[inline,shortlabels]{enumitem}
\usepackage{bbm}
\usepackage[usenames,dvipsnames]{color}
\usepackage[colorlinks=true,linkcolor=blue, citecolor=blue, urlcolor=blue]{hyperref}
\usepackage{fancyhdr}
\usepackage{epsfig}
\usepackage{epstopdf}
\usepackage{caption}
\usepackage[textfont=footnotesize,justification=centering]{subcaption}
\usepackage{multirow}
\usepackage{appendix}
\usepackage{tabularx}
\usepackage{threeparttable}
\usepackage{float}
\usepackage{color}

\usepackage{titling}
\usepackage{sectsty}
\subsubsectionfont{\normalfont\itshape}

\newtheorem*{theorem-non}{Theorem}
\newtheorem{theorem}{Theorem}
\newtheorem{definition}{Definition}
\newtheorem{proposition}{Proposition}
\newtheorem{lemma}{Lemma}
\newtheorem{corollary}{Corollary}
\newtheorem{remark}{Remark}

\providecommand{\keywords}[1]
{
  	
  \textbf{\textit{Keywords ---}} #1
}

\usepackage[UTF8,scheme=plain]{ctex}

\def \E{\mathbb{E}}
\def \P{\mathbb{P}}

\DeclareMathOperator*{\esssup}{ess\,sup}
\usepackage[utf8]{inputenc}

\title{PreFER: Interactive Robo-Advisor with Scoring Mechanism}
\author{Yuwei Wang\thanks{School of Mathematics, Shanghai University of Finance and Economics, Shanghai 200433, China. Email: \href{mailto:wangyuwei@shufe.edu.cn}{wangyuwei@shufe.edu.cn}}
\and Hoi Ying Wong\thanks{Corresponding author. Department of Statistics and Data Science, The Chinese University of Hong Kong, Shatin, N.T., Hong Kong. Email: \href{mailto: hywong@cuhk.edu.hk}{hywong@cuhk.edu.hk}}}

\date{}

\begin{document}

\maketitle
\vspace{-6em}
\begin{abstract}
We propose an interactive robo-advising framework that learns personalized risk preferences from scores provided by clients. The resulting preference-learning problem is closely related to inverse reinforcement learning (IRL), as the robo-advisor infers the client's latent reward specification from feedback. The robo-advisor interacts with clients iteratively as follows. At each interaction time, the advisor generates investment advice based on the optimal policy distribution derived from an inferred personalized risk preference. The client scores the advice. 
The advisor updates its assessment of the client's risk preference based on the feedback.
This learning procedure motivates us to investigate discrete-time Predictable Forward Exploratory Reward (PreFER) processes and derive an exploratory investment strategy. By interpreting the score as the acceptance probability of a piece of advice, our inverse learning procedure learns the client's exploratory investment distribution using the acceptance-rejection method pioneered by von Neumann. Under CARA preferences, we show that, even though the scores contain noise, the robo-advisor can identify the client's current risk aversion after a sufficiently large number of interactions. The PreFER process then carries the learned preference forward and generates future recommendations under updated market conditions.

\end{abstract}


\newpage
\section{Introduction}\label{sec:Introduction}
Robo-advisors are algorithm-driven wealth managers that provide personalized financial advice and automated investment management services \cite{d2021robo}. Learning individual clients' characteristics is a key feature of an effective robo-advisor, with risk preference being one of the most significant factors. Such preferences may evolve over time and are generally difficult for clients to articulate through a single self-reported risk-aversion parameter.

To motivate our framework, consider how a human financial advisor typically interacts with a client. 
The advisor first forms a view of the current market and communicates this view to the client, based on which an investment product or portfolio is recommended. 
The client can then evaluate whether the proposed investment prospect is consistent with her own preferences, even before the eventual investment outcome is realized. 
Through repeated interactions, the advisor gradually develops a better understanding of the client's preferences and can provide increasingly personalized recommendations. 
A robo-advisor can automate this interaction process, allowing more frequent communication with clients at a lower cost.

As indicated by \cite{bianchi2024human}, human--robot interactions can effectively promote investors' financial literacy. In this spirit, our framework regards a key role of the robo-advisor as helping clients better understand their own risk profiles.
We are aware of three main approaches in the literature to eliciting personal risk aversion through human interaction. The first approach allows clients to declare their risk aversion directly during the interaction \cite{capponi2022personalized, liang2023predictable}. Self-declaring a single number requires a complete understanding of the meaning of risk aversion for different objective functions and may lead to emotional responses. The second approach learns risk aversion from investors' portfolio choices, and the performance is improved through sequential elicitation, in which the investor is asked to select from a finite number of pre-constructed portfolios \cite{alsabah2021robo, cheng2023eliciting}. 
The third approach \cite{wang2021robo} infers risk aversion from a continuous set using historical portfolio choices and then mitigates potential biases with subjective client surveys. These approaches rely on either self-reported risk aversion or choices from pre-specified alternatives, motivating our scoring-based interaction.

Building upon the aforementioned works, we propose a novel robo-advising mechanism that learns personal risk preference through scoring. This mechanism involves a three-step sequential and iterative learning procedure:
\begin{itemize}
\item[Step 1] At time $t$, based on its current assessment of the client, the robo-advisor recommends an investment prospect that it considers suitable for the client.
\item[Step 2] The client either accepts the recommendation or provides feedback indicating how well the proposed prospect matches her preferences.
\item[Step 3] The robo-advisor uses the feedback to update its assessment of the client's risk preference and generates a new personalized recommendation, repeating the interaction until the client is satisfied. The learned preference is then used to predict the forward utility at time $t+1$.
\end{itemize}

To implement this interactive mechanism, we introduce predictable forward exploratory reward (PreFER) processes by incorporating relaxed stochastic control into predictable forward performance processes (PFPPs). We characterize the associated optimal control distribution and show that its density is maximized at the classical PFPP optimal strategy, while the exploration parameter governs its dispersion. We then develop a scoring-based learning procedure that identifies the client's risk preference under error-free feedback and increases the probability of acceptable recommendations under noisy feedback.

Our learning problem is closely related to inverse reinforcement learning (IRL), as the robo-advisor infers a latent reward specification from client feedback. 
For the robo-advising application considered in this paper, we specialize the PreFER framework to a structured parametric setting with CARA preferences, in which the reward function is characterized by the client's risk-aversion parameter. This structured specification yields a tractable inverse problem and allows an explicit characterization of the learning procedure. While the underlying interactive framework can in principle accommodate richer parametric reward classes, such extensions may no longer admit the one-to-one identification obtained under the CARA specification and are left for future study.
The learning procedure is interactive and adaptive, and a recommendation is deemed acceptable once its score exceeds the client's threshold, without requiring exact recovery of risk aversion. Table \ref{tab:Comparison} presents a comparative overview of how risk aversion is elicited in the robo-advising literature.

\begin{table}[!htb]
  \centering
  \caption{Comparison of risk aversion elicitation in key literature}
  \begin{tabularx}{\textwidth}{l
      >{\raggedright\arraybackslash\setlength\hsize{1.2\hsize}}X
      >{\raggedright\arraybackslash\setlength\hsize{.65\hsize}}X
      >{\raggedright\arraybackslash\setlength\hsize{.65\hsize}}X}
    \hline
     & Required inputs & Risk aversion set & Performance guarantee \\
    \hline
    \cite{capponi2022personalized,liang2023predictable}
      & self-declaring risk aversion
      & continuous
      & N/A \\
    \cite{wang2021robo}
      & historical data
      & continuous
      & N/A \\
    \cite{alsabah2021robo}
      & observed portfolio choices
      & finite
      & $\epsilon$-close \\ 
    \cite{cheng2023eliciting}
      & hypothetical portfolio choices
      & finite
      & convergence of Gibbs measure \\
    This paper
      & scores of recommendations
      & continuous
      & threshold score \\
    \hline
  \end{tabularx}
  \label{tab:Comparison}
\end{table}

Figure \ref{figDiagram} illustrates the process of risk-preference learning and investment recommendation.
In our framework, the robo-advisor first forms and communicates a market outlook, which is treated as given during the subsequent preference-learning interaction. PreFER then learns the client's risk preference conditional on this shared environment rather than learning the market model itself. Market-model misspecification may affect the implementation of the learned preference without necessarily distorting its identification. A robust treatment of this implementation error is beyond the scope of this paper and is left for future study.
\begin{figure}[h]
    \centering
   \includegraphics[width=12cm]{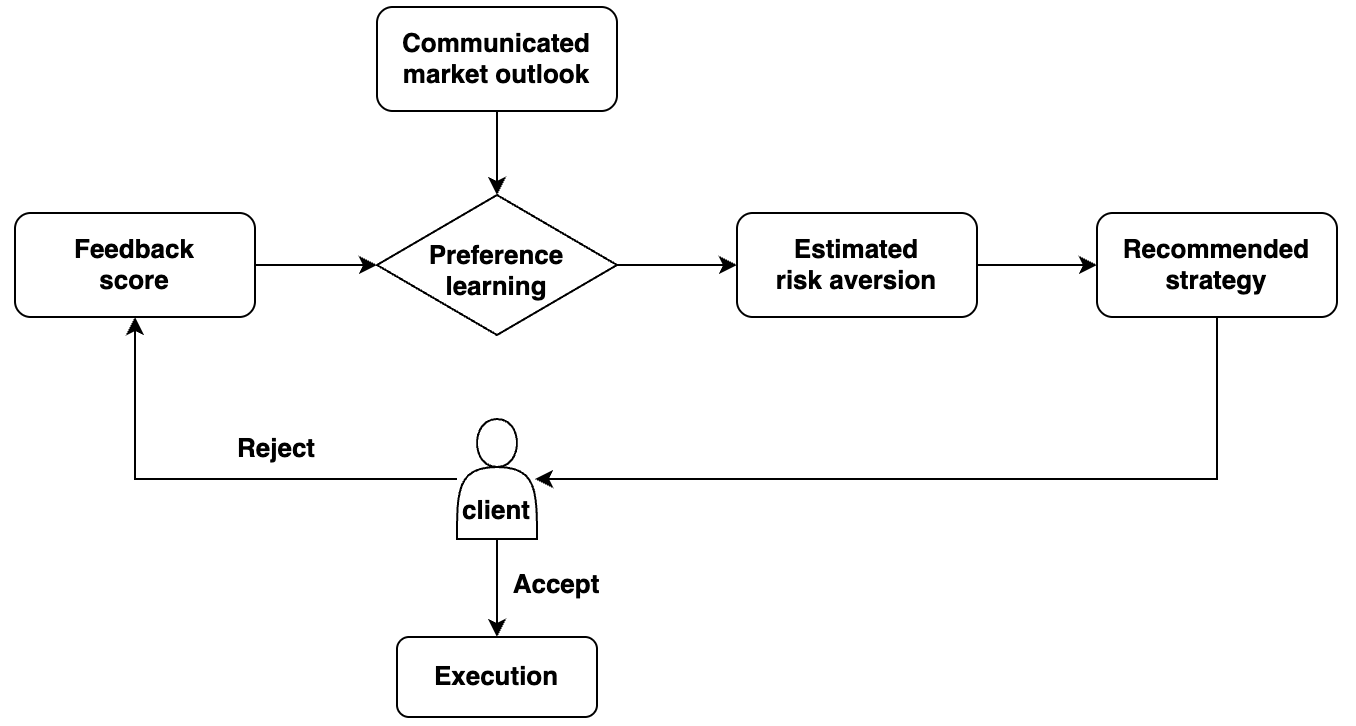}
  \caption{Schematic of the interactive learning process}
  \label{figDiagram}
\end{figure}

To model the portfolio recommendations generated by the robo-advisor, we adopt a relaxed stochastic control formulation, in which admissible controls are represented by probability distributions over the action space. Entropy-regularized relaxed control has been widely used to model exploration in reinforcement learning and related stochastic control problems
\cite{wang2020reinforcement, wang2020continuous, jia2022policya, jia2022policyb, jia2023q, guo2022entropy, dai2023learning, wu2024reinforcement, gao2022state, cui2023discrete}.
Most existing formulations, however, are developed within expected-utility frameworks with a fixed investment horizon and a terminal preference criterion. This is less suited to our setting, where the client's current risk preference is learned through interaction and then needs to be carried forward consistently over time, as described in Step 3.

Forward performance processes provide a natural framework for this purpose, as they avoid ex ante horizon specification and allow dynamic adaptation to changes in both investor preferences and market conditions. Since their introduction in various forms by
\cite{musiela2006investments, musiela2008optimal, musiela2009portfolio, musiela2010portfolio, henderson2007horizon, choulli2007minimal, vzitkovic2009dual},
forward performance processes have been developed in a wide range of settings
\cite{kallblad2018dynamically, källblad2020black, hu2020systems, he2021forward, Bo2022, mohamed2013exact, el2014, el2018consistent, Avanesyan2020, nadtochiy2017optimal, shkolnikov2016asymptotic, chong2019ergodic, choulli2017explicit, nadtochiy2014class, liang2017representation}.
Their discrete-time counterpart, predictable forward performance processes (PFPPs), was introduced in \cite{angoshtari2020predictable} and further studied in
\cite{strub2021evolution, liang2023predictable, liang2023predictableRelative, angoshtari2022predictable, Angoshtari2024, waldon2024forward}.

As established in \cite{liang2023predictable}, PFPPs provide a natural basis for robo-advising, allowing flexible investment horizons and dynamic adaptation to changing market conditions. In our setting, the forward framework further allows information learned from past interactions to be carried forward over time and used to generate future recommendations that remain consistent with the client's learned preference. These recommendations can adapt to evolving market conditions without requiring repeated human involvement, while new client feedback at subsequent interaction times can be incorporated to further refine the learned preference.

Specifically, we consider exploratory control in discrete time and continuous action space in our PreFER processes, where feasible controls are randomized and represented by distributions. 
By employing distributional control, our approach explicitly models the client's willingness to consider a spectrum of strategy recommendations. The client need not express her preference through an explicit risk-aversion parameter, but can instead evaluate whether a presented investment prospect is too aggressive, too conservative, or otherwise consistent with her preferences. This naturally motivates the scoring mechanism, through which clients can express varying degrees of satisfaction with the recommended investment prospects.
For any given initial utility, we show that the density of the optimal exploratory distribution is maximized at the classical PFPP optimal strategy without exploration, regardless of the intended exploration level, while the exploration parameter governs the dispersion of the distribution. This is consistent with the entropy-regularized exploratory formulation \cite{wang2020reinforcement}, where exploitation and exploration are captured by the mean and variance of a Gaussian distribution.
Between interaction times, the forward framework allows investment recommendations to adapt to evolving market conditions while remaining consistent with the client's learned preference. 
Even when client feedback contains observation errors, under the CMA update we derive a lower bound on the probability of accepting a recommendation that increases with the number of interactions and converges to one.

The score is provided before the realized investment performance is observed. Rather than asking the client to evaluate a technical portfolio allocation, the robo-advisor can present the recommendation as an interpretable investment prospect, for example, in terms of potential gains and losses and their associated probabilities. The client can then assess whether the prospect is consistent with her preferences under the communicated market outlook. This is consistent with the human-advisor interaction described above, where clients can evaluate a proposed investment before its eventual outcome is realized.
The continuous score used in our theoretical formulation represents the intensity of client feedback rather than requiring clients to report a precise numerical value. In practice, such feedback may be elicited through coarse ratings, sliders, ordinal responses, or other noisy signals, while our model explicitly allows for observation errors in client feedback. 
More broadly, scoring-based feedback is widely used in personalized recommendation systems to infer user preferences from explicit signals, such as ratings, or implicit signals, such as clicks and watch time. Platforms such as Amazon, Netflix, and YouTube rely on such feedback to model user preferences and improve recommendation accuracy \cite{covington2016deep, gomez2015netflix}. Satisfaction-oriented recommendation systems have been studied using collaborative filtering, content-based models, and hybrid approaches \cite{ricci2010introduction}, with more recent work using machine learning and reinforcement learning to adapt recommendations dynamically \cite{chen2019top, zhao2018deep}.

In Step 2, we propose an algorithm to learn and update the control distribution using client scores. Within the model, the score is interpreted as the probability that a recommendation sampled from the proposal distribution is accepted by the client. 
Our learning procedure builds on rejection sampling, dating back to von Neumann \cite{von1951variate}. Rejection sampling is a simple and widely used Monte Carlo technique with applications in both theoretical and applied settings
\cite{devroye2006nonuniform,robert1999monte,ripley2009stochastic}.
In standard rejection sampling, samples are generated from a proposal distribution and accepted according to an acceptance probability determined by the target and proposal distributions, so that the accepted samples follow the target distribution. Our procedure reverses this logic: the robo-advisor observes the acceptance information for a recommendation generated from a known proposal distribution and uses it to infer the latent target distribution associated with the client's preference. 
This inversion is tractable because the target and proposal distributions share the same parametric form, with their difference characterized by the client's risk-aversion parameter. More specifically, under the CARA specification studied analytically in this paper, error-free feedback, together with the directional indication of whether the recommendation is too conservative or too aggressive, uniquely identifies the client's risk-aversion parameter. When feedback is subject to errors, the robo-advisor instead seeks recommendations whose scores exceed the client's acceptance threshold, and the resulting approximation accuracy of the inferred risk aversion is controlled by this threshold.

To the best of our knowledge, this is the first robo-advising framework to use scoring-based client feedback to learn risk preferences through sequential interaction. A menu-based approach provides a natural benchmark: if the client could compare the full family of investment prospects and select the preferred one without error, her risk preference could in principle be identified directly. In practice, however, a menu can offer only finitely many alternatives and therefore provides a discretized representation of a continuous risk-preference space. It also requires the client to compare several alternatives simultaneously, and a single menu choice, when affected by noise or inconsistency, may not uniquely reveal the client's underlying risk preference. By contrast, PreFER shifts the inferential burden to the robo-advisor: the client evaluates one prospect at a time, while the robo-advisor sequentially aggregates repeated, potentially noisy feedback and adapts subsequent recommendations.

The organization of the paper is as follows.
Section \ref{sec:ModelSetup} introduces the market model and proposes the discrete-time PreFER process, designed to capture investor satisfaction with portfolio choices. In Section \ref{sec:Construction}, we construct and analyze PreFER processes and their associated optimal distributional strategies under CARA risk preferences.
We develop in Section \ref{sec:InteractiveLearning} an interactive learning procedure that infers the client's risk aversion and identifies acceptable investment recommendations based on client feedback.
Section \ref{sec:conclusion} concludes. Extensions to (i) multiple stocks and (ii) CRRA preferences are relegated to the Appendix, and all proofs are provided there.

\section{Model and definition}\label{sec:ModelSetup}
In this section, we first introduce and formalize the discrete-time predictable forward exploratory reward (PreFER) process.
Let $(\Omega, \mathcal{F}, \P)$ be a probability space.
There is a risk-free bond which does not offer any interest, and a binomial stock 
in the market.
The number of stocks can be increased, as discussed in the appendix.
Throughout the paper, $\mathbb{N}$ denotes the set of positive integers and $\mathbb{N}_0$ the set of nonnegative integers. 

Let $R_t$ denote the total return over period $[t-1, t)$ of the stock. 
We suppose that
$R_t$ is described by a binomial model, $R_t = u_t B_t + d_t (1-B_t), t \in \mathbb{N},$
where $B_t$ takes values in $\{0,1\}$, i.e., $(B_t)_{t \in \mathbb{N}}$ is a sequence of Bernoulli random variables with transition probabilities $(p_t)_{t \in \mathbb{N}}$.  
The market parameters $(d_t)_{t \in \mathbb{N}}$, $(u_t)_{t \in \mathbb{N}}$, and $(p_t)_{t \in \mathbb{N}}$ are allowed to be stochastic processes.
In order to assure absence of arbitrage, we assume that $0 < d_t  < 1 < u_t$ and that $0 < p_t < 1$ for any $t \in \mathbb{N}$. We further assume that the market considered herein offers positive expected excess return, $\E[R_t-1\vert \mathcal{F}_{t-1}]>0$.


We define the filtration $\mathbb{F} = (\mathcal{F}_t)_{t \in \mathbb{N}_0}$ by letting $\mathcal{F}_t$ be the augmented $\sigma$-algebra generated by $(B_j)_{j=1}^t$, and the market parameters $(u_j, d_j, p_j)_{j=1}^{t+1}$.
Note that we have $\mathbb P[B_t=1\mid\mathcal F_{t-1}]=p_t$, and $u_t,d_t,p_t$ are predictable. In the robo-advising interpretation, these parameters represent the market outlook adopted and communicated by the robo-advisor at the beginning of each trading period and are treated as given during the corresponding preference-learning interaction. The outcomes of the stock returns are only known at the end of each trading period.

{
Denote the trading strategies by
$f = \left( f_t \right)_{t \in \mathbb{N}}$, where $f_t$ represents the dollar amount invested in the risky asset over trading period $[t-1, t)$. 
A portfolio is constructed by following the trading strategy on the stock while investing all the remaining wealth in the risk-free bond. 
Given a self-financing trading strategy $f$, the wealth process $X^{f} = \left(X_t^{f}  \right)_{t \in \mathbb{N}}$ evolves according to $X_t^{f} = x + \sum\nolimits_{k=1}^tf_k(R_k-1)$ with initial endowment $X_0=x \in \mathbb{R}$. 
In the framework of discrete-time predictable forward performance processes introduced by \cite{angoshtari2020predictable},
all that is required for a trading strategy $f$ and the associated wealth process $X^{f}$ to be \textit{admissible} is that $f_t$ is real-valued and $\mathcal{F}_{t-1}$-measurable. 
In contrast, a key feature of our model is the incorporation of the relaxed control and an entropy regularization term into the stochastic control framework. 
In the robo-advising interpretation, this randomized control represents exploration over candidate investment recommendations while the client's preference is not yet fully learned.
Within this framework, the relaxed control process $f^{\pi} = \left( f_t^{\pi} \right)_{t \in \mathbb{N}}$
is represented as a measure-valued or distributional control process, characterized by the sequence of density functions \(\pi = \{\pi_t, t \in \mathbb{N}\}\). The wealth dynamics becomes $X_t^{\pi} = x + \sum\nolimits_{k=1}^tf_k^{\pi}(R_k-1)$.

At the beginning of each trading period $[t-1,t)$, an interval $[\underline{\pi}_t,\overline{\pi}_t]$ is specified to represent the uncertainty range of an investor's portfolio allocation. 
Let \(\pi_t^0\) denote the optimal control under the predictable forward process in the absence of exploration.
We assume that $\underline{\pi}_t<\pi_t^0<\overline{\pi}_t$, and denote by $B(\underline{\pi}_t,\overline{\pi}_t)$ the
Borel algebra on $[\underline{\pi}_t,\overline{\pi}_t]$.
Trading constraints can be imposed, for instance, \(\underline{\pi}_t = -\frac{X_{t-1}}{u_t - 1}\) and \(\overline{\pi}_t = -\frac{X_{t-1}}{1 - d_t}\) ensure non-negative wealth.

An admissible control is a process \(\pi = \left( \pi_t \right)_{t \in \mathbb{N}}\) such that \(\pi_t\) is a density function of probability measures on \([\underline{\pi}_t, \overline{\pi}_t]\), absolutely continuous with respect to the Lebesgue measure almost surely for each \(t \in \mathbb{N}\). Additionally, for any \(A \in B([\underline{\pi}_t, \overline{\pi}_t])\), \((\int_A \pi_t(u) du)_{t \in \mathbb{N}}\) is \(\mathcal{F}_{t-1}\)-measurable.  
The set of admissible distributions for trading strategies starting at time \(t\) is denoted by \(\mathcal{A}(t)\), and the corresponding wealth processes, starting from \(X_t = x\), is denoted by \(\mathcal{X}(t, x)\). For brevity, \(\mathcal{A}(0)\) and \(\mathcal{X}(0, x)\) are written as \(\mathcal{A}\) and \(\mathcal{X}(x)\). To simplify notation, the explicit dependence of wealth processes on strategies is often omitted, and \(X \in \mathcal{X}(t, x)\) is used instead.

Since the control at each trading time is represented by a probability distribution over the admissible strategy interval, maximizing the expected utility alone would generally result in a degenerate distribution that places all probability mass on the strategy with the highest expected utility. To encourage exploration and prevent the control distribution from degenerating into a point mass, we incorporate an entropy regularization term into the objective.
In order to encourage and capture the overall level of model exploration and uncertainty in terms of $\pi_t, t\in \mathbb{N}$, we consider in this paper Tsallis entropy, 
$H(\pi_t)=\frac{1}{m-1}(1-\int_{\underline{\pi}_t}^{\overline{\pi}_t}\pi_t^m(s)ds)$,
where $m$ is referred to as Tsallis entropy index characterizing the degree of nonlinearity. When $m\rightarrow 1,$ the Tsallis entropy reduces to the Shannon entropy. For example, we choose $m=2$, then $H(\pi_t)=1-\int_{\underline{\pi}_t}^{\overline{\pi}_t}\pi_t^2(s)ds$. One can equivalently consider maximizing $H(\pi_t)=-\int_{\underline{\pi}_t}^{\overline{\pi}_t}\pi_t^2(s)ds$. 

The investor starts at time \(t = 0\) with preferences over her wealth described by a utility function \(U_0\). We assume that any utility function \(U: \mathbb{R} \to \mathbb{R}\) is twice continuously differentiable, strictly increasing, strictly concave, and satisfies the Inada conditions.  
Given any admissible distributional control $\pi \in \mathcal{A}$, let $$\E_{\pi}\left[U(X_t)\vert \mathcal{F}_{t}\right]:=\int_{\underline{\pi}_1}^{\overline{\pi}_1}\dots\int_{\underline{\pi}_t}^{\overline{\pi}_t}
U(x+\sum\limits_{i=1}^ts_{i}(R_i-1))\pi_1(s_1)\dots\pi_t(s_t)ds_1\dots ds_t.$$
Although the framework and the following definition accommodate general utility functions, our primary results focus on the case where the initial preferences belong to the CARA or CRRA families.
This allows for greater mathematical tractability.

\begin{definition}\label{def:ExploratoryFrowardPreferences}
A family of random functions $\left\lbrace U_{t}: \mathbb{R} \times \Omega \rightarrow \mathbb{R} \vert t \in \mathbb{N}_0 \right\rbrace$ together with the associated family of temperature parameters $\left\lbrace \lambda_{t}: \mathbb{R} \times \Omega \rightarrow \mathbb{R}_{+} \vert t \in \mathbb{N} \right\rbrace$, is called a discrete-time predictable forward exploratory reward (PreFER) process under entropy $H:\mathcal{A}\rightarrow \mathbb{R}$ if it satisfies the following conditions:
\begin{enumerate}
\item[(i)] $U_0\left(x,\cdot\right)$ is constant, $U_{t}\left(x,\cdot \right)$ and $\lambda_{t}\left(x,\cdot \right)$ are $\mathcal{F}_{t-1}$-measurable for each $x\in\mathbb{R}$ and $t \in \mathbb{N}$. 
\item[(ii)] $U_{t}(\cdot ,\omega)$ is a utility function for almost all $\omega \in \Omega$ and all $t \in \mathbb{N}_0$.
\item[(iii)] For any initial wealth $x$ and any $\pi \in \mathcal{A}$, the resulting wealth process $X \in \mathcal{X}(x)$ satisfies
\begin{align*}
\E_{\pi}\left[U_{t-1}(X_{t-1})\vert \mathcal{F}_{t-1}\right] \geqslant \E_{\pi} \left[ U_{t}\left( X_{t} \right)+\lambda_t H(\pi_t) \big\vert \mathcal{F}_{t-1} \right], \quad t \in \mathbb{N}.
\end{align*}
\item[(iv)] For any initial wealth $x$, there exists control distribution $\pi^{*}\in \mathcal{A}$ such that the resulting wealth process $X^{*}$ satisfies
\begin{align*}
\E_{\pi^*}\left[U_{t-1}(X_{t-1}^*)\vert \mathcal{F}_{t-1}\right] = \E_{\pi^*} \left[ U_{t} \left(X^{*}_{t} \right)+\lambda_t H(\pi_t^{*}) \big\vert \mathcal{F}_{t-1} \right], \quad t \in \mathbb{N}.
\end{align*}
\end{enumerate}
\end{definition}

According to Definition \ref{def:ExploratoryFrowardPreferences}, the investment strategy is assessed using the PreFER process, a stochastic framework that emphasizes consistency in performance over time. 
Analogously to the non-exploratory case, conditions (iii) and (iv) require that after integrating against the investment strategy \((f_n^{\pi})_{n=1}^t\), whose density function is given by \((\pi_n)_{n=1}^t\), across all exploration intervals \(\{[\underline{\pi}_n,\overline{\pi}_n]\}_{n=1}^t\) and thus excluding the randomness arising from the control distribution, the entropy-regularized processes  
are a supermartingale and a martingale, respectively, with respect to the filtration \(\mathcal{F}_n\).
A strategy is considered optimal if its wealth process and accumulated entropy dynamically preserve expected performance. Specifically, the strategy's conditional performance at any future time, given the current information, must be equal to its performance today. 
Strategies that fail to satisfy this criterion are deemed suboptimal, consistent with the martingale optimality principle (MOP).
However, since the cumulated entropy term \(\sum\limits_{i=1}^{t} \lambda_i H(\pi_i)\) is \(\mathcal{F}_{t-1}\)-measurable and appears on both sides of the (in)equality of MOP, it cancels out, reducing the expressions to Properties (iii) and (iv) in Definition \ref{def:ExploratoryFrowardPreferences}.

Given wealth $x$ at time $t$, we introduce the set of admissible strategies for single period $[t,t+1)$ by $\mathcal{A}_{t,t+1}(x)$, the corresponding set of wealth at time $t+1$ by $\mathcal{X}_{t,t+1}(x)$.
Properties $(iii)$ and $(iv)$ in Definition \ref{def:ExploratoryFrowardPreferences} imply that
\begin{align}\label{eq:InverseEqn_Dynamic}
\E_{\pi} \left[U_{t-1} \left( X^{*}_{t-1} \right)\vert \mathcal{F}_{t-1}\right] = \esssup_{X_{t} \in \mathcal{X}_{t-1,t} \left( X^{*}_{t-1}\right)} \E_{\pi} \left[ U_{t}  \left(X_{t}\right)+\lambda_t H(\pi_t)  \bigg\vert \mathcal{F}_{t-1}\right].
\end{align}
By iteratively solving \eqref{eq:InverseEqn_Dynamic}, the construction of PreFER processes can be achieved. For a comprehensive discussion on the case with 
$\lambda_t=0$, refer to \cite{angoshtari2020predictable}.
The crucial step is then to solve the following \textit{inverse investment problem with relaxed control}:
Given a profile of initial preferences $U_0$, we seek to determine forward exploratory utility $U_1$, and density function for the investment strategy $\pi_1^{*}$ such that
\begin{equation}\label{eq:InverseProblem}
\begin{aligned}
    U_{0} \left( x \right) = \sup_{X_{1} \in \mathcal{X}_{0,1} \left(x\right)} \E_{\pi} \left[ U_{1}  \left(X_{1}\right)+\lambda_1 H(\pi_1) \right]=\sup_{{\pi_1} \in \mathcal{A}_{0,1}(x)} \E _{\pi}\left[ U_1 \left( x + f_1(R_1-1)\right)+\lambda_1 H(\pi_1) \right]
\end{aligned}
\end{equation}
for any initial wealth $x\in \mathbb{R}$, where $f_1$ is a random control whose probability density is $\pi_1$.
Subsequently, one can iteratively construct $U_{2}$, $U_{3}$, and so forth.
This is achieved by repeatedly solving 
a problem of the form \eqref{eq:InverseProblem}, conditionally on updated information available at the next time point, and arguing that this solution satisfies the required measurability conditions.

Note that the market is sequentially complete across each trading period, assuming full knowledge of the model parameters at the start of the period. Therefore, the equivalent martingale measure for the truncated model of a single period is unique. 
We denote it by $\mathbb{Q}$ and denote by $q_t = \frac{1-d_t}{u_t - d_t}$ the risk-neutral probability for an upward move of the stock in period $[t-1,t), t\in \mathbb{N}$.
The positivity of the expected excess return is then equivalent to $\log \frac{p_t(1-q_t)}{(1-p_t)q_t}>0$. 

Let \(f^{0} = (f_t^{0})_{t \in \mathbb{N}}\) denote the optimal strategy and \((X_t^{0})_{t \in \mathbb{N}}\) the corresponding wealth process, under the predictable forward performance process without exploration. We refer to this as the benchmark strategy. Recall from \cite{angoshtari2020predictable} that \(f^{0}\) follows the Dirac measure \(\pi^{0}\) solving
\begin{equation}\label{PFPPOptimalControl}
\begin{aligned}
    U_{t-1} \left( X_{t-1} \right) = \esssup_{\pi_t \in \mathcal{A}(t-1,X_{t-1})} \E \left[ U_t \left( X_{t-1} + f_t^0(R_t-1)\right)\vert \mathcal{F}_{t-1}\right].
\end{aligned}
\end{equation}
Suppose that 
$\pi^*(s)\in\mathcal{A}$ is the density function of the optimal control $f^*$ such that 
\begin{equation}\label{PreFEROptimalControl}
\begin{aligned}
    \E_{\pi^*} \left[U_{t-1} \left( X_{t-1} \right) \vert \mathcal{F}_{t-1}\right]=\E_{\pi^*} \left[ U_t \left( X_{t-1} + f_t^*(R_t-1)\right)-\lambda_t \int_{\underline{\pi}_t}^{\overline{\pi}_t}(\pi_t^*)^2(s)ds\bigg\vert \mathcal{F}_{t-1}\right].
\end{aligned}
\end{equation}
We establish the following property for the density function of the optimal strategy when constructing the PreFER processes with distributional control.

\begin{proposition}\label{Prop:MaxControlDist}
    For any trading period \([t-1, t)\), the probability density function \(\pi_t^{*}(s)\) of the associated optimal control for constructing the PreFER process, defined over \( s \in [\underline{\pi}_t, \overline{\pi}_t] \), attains its maximum at \( s^* = \pi_t^{0} \).
\end{proposition}

\section{The construction of the PreFER processes}\label{sec:Construction}
In this section, we propose an algorithm for constructing the PreFER processes under Tsallis entropy with \(m = 2\), based on CARA risk preferences. The formulation and results extend to (i) CRRA preferences and (ii) multi-stock binomial models under CARA preferences; full details are relegated to the Appendix.

Consider a criterion of exponential type and construct the associated PreFER processes with initial preferences given by $U_0(x)=-e^{-\gamma x}$, where the constant $\gamma>0$ denotes the coefficient of risk aversion.
To ensure that the forward pair \((U_{t-1}, U_t)\), satisfying the relationship \eqref{eq:InverseEqn_Dynamic}, remains within the class of CARA utilities, we adopt a state-dependent specification of the temperature parameter, thereby introducing an additional degree of freedom that facilitates homotheticity and enhances analytical tractability.

For any $\rho_t>0$, define
$P_t=C_{t-1}J_t$ and $C_t=\prod\limits_{n=1}^{t}P_n$ with $C_0=1$ iteratively, where $J_t=\frac{1}{\beta_t-\frac{1}{4\rho_t}\int_{\underline{\pi}_t}^{\overline{\pi}_t}h_t^2(s)ds},$
\begin{align}\label{Expression-h}
    h_t(s)=-C_{t-1}\left(p_te^{-\gamma(u_t-1)s}+(1-p_t)e^{-\gamma(d_t-1)s}\right)+\beta_t,
\end{align}
and 
\begin{align}\label{Expression-beta}
    \beta_t=\frac{2\rho_t-C_{t-1}\left(\frac{p_t}{\gamma(u_t-1)}(e^{-\gamma(u_t-1)\overline{\pi}_t}-e^{-\gamma(u_t-1)\underline{\pi}_t})+\frac{1-p_t}{\gamma(d_t-1)}(e^{-\gamma(d_t-1)\overline{\pi}_t}-e^{-\gamma(d_t-1)\underline{\pi}_t})\right)}{\overline{\pi}_t-\underline{\pi}_t}.
\end{align}
Consider temperature parameter $\lambda_t:=\rho_t P_te^{-\gamma X_{t-1}}$ for control problem over trading period $[t-1,t)$, where the coefficient $\rho_t>0$ is $\mathcal{F}_{t-1}$-measurable and introduced for quantifying the trade-off between exploitation and exploration. 
A higher temperature reflects greater willingness to explore, or greater indecisiveness among candidate portfolio recommendations, resulting in a more dispersed control distribution. In contrast, a lower temperature reflects a stronger preference for recommendations near the benchmark strategy. Thus, $\rho_t$ controls the intensity of exploration and may also be interpreted as the client's degree of indecisiveness over candidate recommendations.

To ensure the well-posedness of the PreFER processes under CARA risk preferences, we assume that for any prescribed exploration coefficient \(\rho_t\), the interval \([ \underline{\pi}_t, \overline{\pi}_t ]\) must be chosen such that the following conditions are satisfied: 
\begin{align}\label{AssumptionForPositive}
    \beta_t - \frac{1}{4\rho_t} \int_{\underline{\pi}_t}^{\overline{\pi}_t} h_t^2(s) \, ds > 0, \quad h_t(\overline{\pi}_t) > 0 \quad \text{and} \quad h_t(\underline{\pi}_t) > 0.
\end{align}
The assumption is consistent with the interpretation of $\rho_t$, a lower \(\rho_t\), corresponding to less indecisiveness among candidate recommendations, suggests that the interval \([\underline{\pi}_t, \overline{\pi}_t]\) should be narrower and remain closer to the benchmark strategy. In contrast, a higher \(\rho_t\) allows for a wider interval of exploratory strategies.

In the following theorem, we present an explicit construction of the PreFER process and the associated optimal relaxed control. 
\begin{theorem}\label{Thm:ExploForwardRL}
For $t\in \mathbb{N}$, the process $U_t(x)=-C_te^{-\gamma x}$, is a PreFER process.
The density function of the optimal strategy $f_t^{*}$ is given by
\begin{align}\label{OptimalStrategyRL}
    \pi_t^{*}(s)=\frac{-C_{t-1}p_te^{-\gamma(u_t-1)s}-C_{t-1}(1-p_t)e^{-\gamma(d_t-1)s}+\beta_t}{2\rho_t},
\end{align}
which generates the corresponding optimal wealth process
\begin{align}\label{OptimalWealthRL}
     X_t^{*}= X_0+\sum_{n=1}^{t}f_n^{*}(R_n-1).
\end{align}
\end{theorem}

\begin{remark}
Although the expressions and notation introduced above may appear involved, each component has a transparent interpretation. From equations \eqref{Expression-h} and \eqref{OptimalStrategyRL}, we obtain \(\pi_t^*(s) = \frac{h_t(s)}{2\rho_t}\), showing that \(h_t(s)\) determines the shape and functional form of the optimal density, while \(\rho_t\) scales its magnitude. Let \(\pi_t^{un}\) denote the maximizer of the objective in \eqref{SupremumEqn} in the Appendix, without the constraint that the probability density function integrates to one, \eqref{Expression-beta} together with \eqref{OptimalStrategyRL} shows that \(\pi_t^{un}(s)=\pi_t^*(s)+\frac{\beta_t}{2\rho_t}\), thus $\frac{\beta_t}{2\rho_t}$ represents the constant difference between the unconstrained solution and the optimal density.
\end{remark}

Forward performance processes with non-exploratory controls in binomial models have been studied in \cite{musiela2016forward} for exponential risk preferences. An interesting finding of Theorem \ref{Thm:ExploForwardRL} is that the optimal strategy allocates a constant dollar amount to the risky asset, regardless of initial wealth, consistent with the forward exponential framework without exploration.

\begin{corollary}\label{Cor:MaxPi}
For the trading period \([t-1, t)\), \(t \in \mathbb{N}\), the density function $\pi_t^*(s), s\in [\underline{\pi}_t, \overline{\pi}_t]$ of the optimal strategy determined by \eqref{OptimalStrategyRL} is strictly increasing for \(s < \pi_t^0\), strictly decreasing for \(s > \pi_t^0\), and achieves its maximum value at \(s = \pi_t^0\), where 
    \begin{align}\label{BenchmarkStrategy}
        \pi_t^0=\frac{1}{\gamma (u_t-d_t)}\log\frac{p_t(1-q_t)}{(1-p_t)q_t},
    \end{align}
and 
\begin{align*}
    \pi_t^{*}(\pi_t^0)=\frac{-C_{t-1}\left(\frac{1-p_t}{1-q_t}\right)^{1-q_t}\left(\frac{p_t}{q_t}\right)^{q_t}+\beta_t}{2\rho_t}.
\end{align*}
Furthermore, \(\pi_t^*(s)\) becomes increasingly concave as \(\rho\) decreases, and \(\pi_t^{*}(\pi_t^0)\) increases with a decrease in \(\rho\).
\end{corollary}

Notice that \(\pi_t^0\), as computed in \eqref{BenchmarkStrategy}, corresponds to the original optimal strategy, which is obtained when \(\rho_t = 0\), without incorporating any uncertainty. Intuitively, in investment practice, the optimal control distribution assigns the highest weight to the optimal strategy recognized in the classical forward framework, as exploitation. The willingness to attempt nearby strategies diminishes as they deviate further from \(\pi_t^0\). 
On the other hand, the increasing concavity of \(\pi_t^*\) as \(\rho_t\) decreases implies that, with a lower value of \(\rho_t\), \(\pi_t^*\) becomes steeper in both its increase and decrease.
As a result, when $\rho$ is lower, 
our model demonstrates a greater tendency to exploit the strategy near the known optimal solution and a reduced willingness to experiment with choices that are substantially different. Such behavior reflects a more cautious approach to portfolio selection as \(\rho_t\) decreases.

\section{The application in robo-advising}\label{sec:InteractiveLearning}

In this section, we study how PreFER processes under CARA preferences can be used to learn a client's risk preference from feedback on recommended investment prospects. The client evaluates each recommendation before its realized investment performance is observed, based on the prospect communicated by the robo-advisor. We model the intensity of this evaluation through a continuous score, while allowing the observed feedback to be subject to error.

The acceptance-rejection method, often referred to as rejection sampling, is a classical Monte Carlo technique for universal sampling. It was first introduced by \cite{von1951variate} and has since become a cornerstone in computational statistics.
It enables generating samples from any desired distribution by drawing a simulation from a simpler proxy distribution. In the acceptance-rejection algorithm for generating continuous random variable $X$ with distribution as $P_X$, a candidate \(Y\) is drawn from the proxy distribution \(P_Y\) along with an independent uniform random variable \(U\) over the unit interval. If \(U\leq \frac{P_X(Y)}{c \, P_Y(Y)}\) for a constant $c\geq\sup_s\{P_X(s)/P_Y(s)\}$, \(Y\) is accepted as a sample drawn from \(X\); otherwise, the candidate is discarded and the procedure is repeated. 
The number of iterations required to generate an acceptable \(X\) follows a geometric distribution, with ``success probability" $\P(U \leq \frac{P_X(Y)}{c\,P_Y(Y)})$. Building on this property, we interpret the acceptance probability as a latent satisfaction score associated with the recommended investment prospect. The score measures how well the prospect aligns with the client's preference under the communicated market outlook, before the realized investment outcome is observed.

Our inverse learning problem is more structured than unrestricted IRL. Under the CARA specification considered in this section, the unknown reward specification is characterized by the scalar risk-aversion parameter $\gamma$. For fixed market parameters and a given recommendation, the acceptance probability, together with whether the client regards the recommendation as too conservative or too aggressive, provides an invertible mapping to the client's risk aversion. This parametric structure therefore enables explicit identification of the preference parameter.

Even when the client's risk preference remains unchanged, the acceptance probability associated with a given recommendation may vary with market conditions. To accommodate such changes, we construct an acceptance interval based on the collected feedback and a threshold score $\alpha$. Any recommendation within this interval is guaranteed to achieve a score of at least $\alpha$.

\subsection{Inferring investor's risk aversion through scoring}
A central challenge in robo-advising is to infer the client's current risk preference, which may evolve over time and cannot generally be elicited through a single precise self-reported parameter. Accurate inference of risk tolerance is therefore important for generating personalized investment recommendations.

Throughout the preference-learning procedure within a fixed trading period, the market parameters communicated by the robo-advisor are held fixed, and only the client's risk-aversion parameter is inferred from feedback.
In this paper, we design a robo-advising framework that utilizes the one-to-one correspondence between the benchmark strategy \(\pi_t^0\) and the underlying risk aversion parameter \(\gamma_t\). For any period $[t-1,t)$ and a prescribed range of possible risk aversion values \([\underline{\gamma_t}, \overline{\gamma_t}]\), the corresponding interval of candidate strategies for exploration is then \([\underline{\pi}_t, \overline{\pi}_t]\) with $\underline{\pi}_t=\frac{1}{\overline{\gamma_t} (u_t - d_t)} \log \frac{p_t(1-q_t)}{(1-p_t)q_t}$ and $\overline{\pi}_t=\frac{1}{\underline{\gamma_t} (u_t - d_t)} \log \frac{p_t(1-q_t)}{(1-p_t)q_t}$, which readily follows from standard optimization arguments.

Assume that the pair of $\underline{\gamma_t}$ and $\overline{\gamma_t}$ is chosen such that 
\begin{align}\label{Assum:Lemma1}
    \frac{1- (1-q_t)\frac{\underline{\gamma_t}}{\overline{\gamma_t}}\log \frac{p_t(1-q_t)}{(1-p_t)q_t}}{1+ q_t\frac{\underline{\gamma}_t}{\overline{\gamma_t}}\log \frac{p_t(1-q_t)}{(1-p_t)q_t}}\left(\frac{p_t(1-q_t)}{(1-p_t)q_t}\right)^{(1-\frac{\underline{\gamma}_t}{\overline{\gamma_t}})}<1.
\end{align}
We remark that a larger exploration parameter $\rho$ corresponds to greater willingness to explore, or greater indecisiveness among candidate recommendations, and therefore leads to a wider exploration range.
At each interaction, the robo-advisor proposes the maximizer of its estimated distributional control, that is, the strategy that would be accepted with probability one under the hypothetical condition that the estimate is correct. 
The score assigned by the client to the corresponding investment prospect is then incorporated to update the robo-advisor's inference of the client's underlying risk aversion in subsequent iterations.

Consider any period $[t-1,t)$, let \(\gamma^a\) denote the client's actual risk aversion, and \(\gamma^{(n)}\) the robo-advisor's $n$-th estimate. Suppose that the utility function of the client at time $t-1$ is given by $U_{t-1}(x)=-C_{t-1}e^{-\gamma^a x}$.
The probability density functions of the optimal controls associated with these parameters are denoted by \(\pi_t^{*,a}(s)\) and \(\pi_t^{*,(n)}(s)\), respectively. 
The strategy corresponding to the client's true risk aversion and that generated under the robo-advisor's $n$-th estimate \(\gamma^{(n)}\) are respectively given by 
\[ 
S_{\gamma^{a}} = \frac{1}{\gamma^{a} (u_t - d_t)} \log \frac{p_t(1-q_t)}{(1-p_t)q_t}, \qquad
S_{\gamma^{(n)}} = \frac{1}{\gamma^{(n)} (u_t - d_t)} \log \frac{p_t(1-q_t)}{(1-p_t)q_t}.
\]  
Building on the acceptance-rejection method, we define the acceptance probability as a satisfaction score associated with the investment prospect generated by the proposed risk-aversion estimate. This quantity reflects the client's evaluation of the recommended prospect relative to her actual preferences.
\begin{definition}\label{Def:AcceptanceProbability}
    Consider any trading period $[t-1,t)$, $t\in\mathbb N$. The acceptance probability associated with the investment recommendation generated from the estimate $\gamma^{(n)}$ is defined by $\frac{\pi_t^{*,a}(S_{\gamma^{(n)}})}{c_t^{(n)} \pi_t^{*,(n)}(S_{\gamma^{(n)}})},$ where $c_t^{(n)} = \sup_{s \in [\underline{\pi}_t, \overline{\pi}_t]} \frac{\pi_t^{*,a}(s)}{\pi_t^{*,(n)}(s)}.$
\end{definition}
\begin{remark}
The continuous score in Definition \ref{Def:AcceptanceProbability} is a theoretical representation of the intensity of client feedback. It does not require the client to know or report her risk-aversion parameter, nor to communicate an exact acceptance probability. In practice, the observed feedback may be coarser or noisy; Section~\ref{Sec:InteractionProcedures} models the resulting observation error in the preference inferred from such feedback.
\end{remark}

In the definition, we choose $c_t^{(n)}$ as the minimal constant that introduces the least overhead while ensuring a full score when the inferred risk aversion coincides with the client's true risk aversion.
To investigate the property of acceptance probability,
we present the following auxiliary result for
the monotonicity of $(\pi_t^*)'(s)$, the first derivative in $s$, in risk aversion $\gamma$. 
\begin{lemma}\label{lemma:MonoPiPrime}
    For any fixed $s\in[\underline{\pi}_t, \overline{\pi}_t]$,  $(\pi_t^*)'(s)$ decreases strictly as $\gamma$ increases. Moreover, a lower value of $\gamma$ results in $\pi_t^*(\underline{\pi}_t)$ being smaller, while $\pi_t^*(\overline{\pi}_t)$ becomes larger. The equation $\pi_t^{*,a}(s)=\pi_t^{*,(n)}(s)$ admits a unique solution on the interval $[\underline{\pi}_t, \overline{\pi}_t]$, given $\gamma^{*,a}\neq \gamma^{*,(n)}$.
\end{lemma}

We know from Lemma \ref{lemma:MonoPiPrime} that a smaller $\gamma$ reduces the density of optimal control at $\underline{\pi}$ while increasing it at $\overline{\pi}$, thereby lowering the weight on high–risk‐aversion strategies and favoring low–risk‐aversion ones. The optimal densities under any two unequal risk aversion parameters intersect exactly once.
Consider $\beta_t$ given by \eqref{Expression-beta} as a function of risk aversion $\gamma$, denoted by $\beta_t(\gamma)$, we define $\beta_t(\gamma^{a})=\beta_t^{a}$ and $\beta_t(\gamma^{(n)})=\beta_t^{(n)}$. 
The following proposition establishes that the supremum \(c_t^{(n)}\) is attained at one of the boundary points, $\underline{\pi}_t$ or $\overline{\pi}_t$, depending on whether the recommendation generated from $\gamma^{(n)}$ made by the robo-advisor is less or greater than the actual value.
\begin{proposition}\label{Prop: AcceptanceProb.c}
    If $\gamma^a<\gamma^{(n)}$, then the supremum of $\frac{\pi_t^{*,a}(s)}{\pi_t^{*,(n)}(s)}$ is attained at $\overline{\pi}_t$, and $$c_t^{(n)}=\frac{-C_{t-1}p_te^{-\gamma^a(u_t-1)\overline{\pi}_t}-C_{t-1}(1-p_t)e^{-\gamma^a(d_t-1)\overline{\pi}_t}+\beta_t^a}{-C_{t-1}p_te^{-\gamma^{(n)}(u_t-1)\overline{\pi}_t}-C_{t-1}(1-p_t)e^{-\gamma^{(n)}(d_t-1)\overline{\pi}_t}+\beta_t^{(n)}}.$$ 
    Otherwise, the supremum of $\frac{\pi_t^{*,a}(s)}{\pi_t^{*,(n)}(s)}$ is attained at the lower bound $\underline{\pi}_t$, and $$c_t^{(n)}=\frac{-C_{t-1}p_te^{-\gamma^a(u_t-1)\underline{\pi}_t}-C_{t-1}(1-p_t)e^{-\gamma^a(d_t-1)\underline{\pi}_t}+\beta_t^a}{-C_{t-1}p_te^{-\gamma^{(n)}(u_t-1)\underline{\pi}_t}-C_{t-1}(1-p_t)e^{-\gamma^{(n)}(d_t-1)\underline{\pi}_t}+\beta_t^{(n)}}.$$ 
\end{proposition}

Moreover, since 
\begin{align*}
    \frac{\pi_t^{*,a}(S_{\gamma^{(n)}})}{\pi_t^{*,(n)}(S_{\gamma^{(n)}})}=\frac{-C_{t-1}p_t\left(\frac{p_t(1-q_t)}{(1-p_t)q_t}\right)^{-\frac{\gamma^{a}}{\gamma^{(n)}}(1-q_t)}-C_{t-1}(1-p_t)
    \left(\frac{p_t(1-q_t)}{(1-p_t)q_t}\right)^{\frac{\gamma^{a}}{\gamma^{(n)}}q_t}+\beta_t^{a}}{-C_{t-1}\left(\frac{1-p_t}{1-q_t}\right)^{1-q_t}\left(\frac{p_t}{q_t}\right)^{q_t}+\beta_t^{(n)}},
\end{align*}
for instance when $\gamma^{a}<\gamma^{(n)}$, 
it follows from Proposition \ref{Prop: AcceptanceProb.c} that the acceptance probability for the recommended strategy $S_{\gamma^{(n)}}$ can be expressed by $\frac{\pi^{*,a}(S_{\gamma^{(n)}})}{c_t^{(n)}\pi^{*,(n)}(S_{\gamma^{(n)}})}=\frac{g_1(\gamma^{(n)},\gamma^{a})}{g_2(\gamma^{a})}h(\gamma^{(n)})$, where
\begin{align*}
    g_1(\gamma^{(n)},\gamma^{a})=-C_{t-1}p_t\left(\frac{p_t(1-q_t)}{(1-p_t)q_t}\right)^{-\frac{\gamma^{a}}{\gamma^{(n)}}(1-q_t)}-C_{t-1}(1-p_t)\left(\frac{p_t(1-q_t)}{(1-p_t)q_t}\right)^{\frac{\gamma^{a}}{\gamma^{(n)}}q_t}+\beta_t^{a},
\end{align*}
\begin{align*}
    g_2(\gamma^{a})=-C_{t-1}p_te^{-\gamma^a(u_t-1)\overline{\pi}_t}-C_{t-1}(1-p_t)e^{-\gamma^a(d_t-1)\overline{\pi}_t}+\beta_t^a,
\end{align*}
and 
\begin{align*}
    h(\gamma^{(n)})=\frac{-C_{t-1}p_te^{-\gamma^{(n)}(u_t-1)\overline{\pi}_t}-C_{t-1}(1-p_t)e^{-\gamma^{(n)}(d_t-1)\overline{\pi}_t}+\beta_t^{(n)}}{-C_{t-1}\left(\frac{1-p_t}{1-q_t}\right)^{1-q_t}\left(\frac{p_t}{q_t}\right)^{q_t}+\beta_t^{(n)}}.
\end{align*}
An interesting result is that the acceptance probability is strictly decreasing as the guess $\gamma^{(n)}$ deviates further from $\gamma^{a}$, or as $\gamma^{a}$ deviates further from the guess $\gamma^{(n)}$.
We establish this property in the following proposition. 
\begin{proposition}\label{Prop:MonotonicityAcceptProb}
    For any fixed $\gamma^{a}\in[\underline{\gamma_t},\overline{\gamma_t}]$, $\frac{\pi^{*,a}(S_{\gamma^{(n)}})}{c_t^{(n)}\pi^{*,(n)}(S_{\gamma^{(n)}})}$ is increasing in $\gamma^{(n)}\in[\underline{\gamma_t},\gamma^{a})$ while decreasing in $\gamma^{(n)}\in[\gamma^{a},\overline{\gamma_t})$, and attains its maximum at $\gamma^{(n)}=\gamma^{a}$. For any fixed $\gamma^{(n)}\in[\underline{\gamma_t},\overline{\gamma_t}]$, $\frac{\pi^{*,a}(S_{\gamma^{(n)}})}{c_t^{(n)}\pi^{*,(n)}(S_{\gamma^{(n)}})}$ is increasing on $\gamma^{a}\in[\underline{\gamma_t},\gamma^{(n)})$ while decreasing on $\gamma^{a}\in[\gamma^{(n)},\overline{\gamma_t})$, and attains its maximum at $\gamma^{a}=\gamma^{(n)}$.
\end{proposition}

This proposition justifies that our scoring scheme based on acceptance probability is well-defined: When the $n$-th guess $\gamma^{(n)}$ deviates further from the client's actual risk aversion $\gamma^a$, the score is strictly lower. If the client specifies a threshold score, it effectively defines a single interval for the candidate guess $\gamma^{(n)}$. On the other hand, from the perspective of computing and inferring $\gamma^a$ when the robo-advisor receives the  score as client feedback on the fixed guess $\gamma^{(n)}$, Proposition \ref{Prop:MonotonicityAcceptProb} shows that for any fixed guess $\gamma^{(n)}$, the acceptance probability as score is strictly decreasing as the actual $\gamma^{a}$ is further from $\gamma^{(n)}$.

For each fixed $\gamma^a$ and threshold score $\alpha$,
denote for simpler notation 
$F_{\gamma^a}(\gamma^{(n)})=\frac{\pi_t^{*,a}(S_{\gamma^{(n)}})}{c_t^{(n)} \pi_t^{*,(n)}(S_{\gamma^{(n)}})}$ for the $n$-th guess $\gamma^{(n)}\in(0,+\infty)$. We consider the domain separately: $F_{\gamma^a}^{L}(\gamma^{(n)})=F_{\gamma^a}(\gamma^{(n)})$ for $\gamma^{(n)}\in(0,\gamma^a)$, and $F_{\gamma^a}^{R}(\gamma^{(n)})=F_{\gamma^a}(\gamma^{(n)})$ for $\gamma^{(n)}\in(\gamma^a,+\infty)$. By Lemma \ref{lemma:MonoPiPrime} the inverse of $F_{\gamma^a}^{L}(\gamma^{(n)})$ and $F_{\gamma^a}^{R}(\gamma^{(n)})$ exist, which we denote by $(F_{\gamma^a}^{L})^{-1}$ and $(F_{\gamma^a}^{R})^{-1}$.
Therefore, the condition that an investment suggestion generated from the guess \( \gamma^{(n)} \) is accepted is \( F_{\gamma^a}(\gamma^{(n)}) > \alpha \), or equivalently \( (F_{\gamma^a}^L)^{-1}(\alpha) < \gamma^{(n)} < (F_{\gamma^a}^R)^{-1}(\alpha) \). We present the following corollary to demonstrate that 
if the client's potential preferences are known to lie within a given range, one can identify an interval of candidate estimates that is acceptable for every actual preference in that range. 
We refer to such an interval as the acceptance interval.

\begin{corollary}\label{Cor:AcceptableForAll}
    Suppose that $(F_{\gamma_r^a}^L)^{-1}(\alpha)<(F_{\gamma_l^a}^R)^{-1}(\alpha)$, $0<\gamma^a_l<\gamma^a_r<\infty$ and $\alpha\in(0,1)$, then $\frac{\pi_t^{*,a}(S_{\gamma^{(n)}})}{c_t^{(n)} \pi_t^{*,(n)}(S_{\gamma^{(n)}})}\geq \alpha$ for any $\gamma^{(n)}\in[(F_{\gamma_r^a}^L)^{-1}(\alpha),(F_{\gamma_l^a}^R)^{-1}(\alpha)]$ and any $\gamma^{a}\in[\gamma^a_l, \gamma^a_r]$.
\end{corollary}

\subsection{Designing interaction procedures}\label{Sec:InteractionProcedures}

Since a client's risk preference may evolve over time, she may initiate a new interaction when she wishes to reassess her investment preference. 
The robo-advisor then infers her current risk aversion using the scoring mechanism based on the PreFER process.
We omit the time subscript $t$ herein, since the inference of risk aversion is based on investment preferences in a single fixed period when the client initiates an update.
According to Proposition \ref{Prop:MonotonicityAcceptProb}, there exists a one-to-one correspondence, separately on \(\gamma^{(n)}>\gamma^a\) or \(\gamma^{(n)}<\gamma^a\), between the client's feedback and the proposed risk aversion \(\gamma^{(n)}\). Hence, from her feedback to any fixed guess \(\gamma^{(n)}\), the robo-advisor can infer the actual risk aversion, with up to two possible solutions. This inference becomes unique once the client indicates whether the recommendation is perceived as too conservative or too aggressive.

If the client can provide an accurate score and specify whether the recommended strategy is too conservative or too aggressive, then a single interaction suffices to identify \(\gamma^a\). 
In practice, however, inconsistencies in the client's decisions, whether due to mistakes or emotional influences, may hinder the efficiency of the learning process. Recognizing that clients may not be able to articulate their preferences precisely through feedback, the communicated score may fail to perfectly reflect the true risk-aversion parameter.

To capture errors in the observed feedback, let $\gamma^{e,(n)}$ denote the risk-aversion value implied by the client's feedback at the $n$-th interaction. We model this feedback-implied value as
\[
\gamma^{e,(n)}=\gamma^a+\epsilon^{(n)},
\]
where $\{\epsilon^{(n)}\}_{n\geq1}$ is an i.i.d. sequence of real random variables supported on $[-d_\epsilon^L,d_\epsilon^R]$, with zero mean and variance $\sigma_\epsilon^2$. Thus, the client is not assumed to report a risk-aversion parameter directly; rather, $\gamma^{e,(n)}$ represents the parameter value inferred by the robo-advisor from noisy client feedback. The constants $d_\epsilon^L,d_\epsilon^R>0$ represent the degree of noise or inconsistency in the feedback-implied preference.
Assume that $d_{\epsilon}^R<\gamma^a$ (a condition used later to establish convergence), and let \([\underline{\gamma}, \overline{\gamma}]\)  denote the exploratory range of possible risk aversion values. We are only interested in the case where $d_{\epsilon}$ is chosen so that the noise-perturbed risk aversion remains within this exploration interval but may, with positive probability, fall outside the acceptance interval, i.e., 
$\underline{\gamma}\leq \gamma^a-d_{\epsilon}^L< (F^L_{\overline{\gamma}})^{-1}(\alpha)$ and $ (F^R_{\underline{\gamma}})^{-1}(\alpha)<\gamma^a+d_{\epsilon}^R\leq \overline{\gamma}$. Otherwise, the error is acceptable by Corollary \ref{Cor:AcceptableForAll} and a single interaction would suffice to produce a satisfactory parameter estimate.

The interaction procedure within a single period is structured as follows:
\begin{itemize}
    \item Initial Guess and Evaluation:  
   The robo-advisor proposes an initial guess \(\gamma^{(1)}\). The client, whose true risk aversion is \(\gamma^a\), evaluates the recommended strategy. If she is satisfied, the interaction terminates, and the robo-advisor invests using \(\gamma^{(1)}\).
   \item Feedback and Error:  
   If the client is not satisfied, she provides feedback on the recommended prospect. The robo-advisor maps this feedback, through the inverse score relationship, into a feedback-implied risk-aversion value $\gamma^{e,(1)}$, which is modeled as
    \(\gamma^{e,(1)}=\gamma^a+\epsilon^{(1)}.\)
    The robo-advisor then uses $\gamma^{e,(1)}$ to generate an updated guess $\gamma^{(2)}$.
   \item Iterative Update: 
   This interaction continues in the same manner. At each iteration, the client either accepts the recommendation or provides further feedback on the proposed prospect. The robo-advisor converts this feedback into a feedback-implied risk-aversion value
    $\gamma^{e,(n)}=\gamma^a+\epsilon^{(n)}$
    and updates its estimate until the client accepts or a stopping criterion is met.
\end{itemize}

In updating the estimate of the client's risk aversion, we consider two approaches. Approach A is Markovian, uses only the most recent feedback $\gamma^{e,(N-1)}$. 
Approach B employs a cumulative moving average. 
$\gamma^{(N)} = \frac{1}{N-1} \sum_{k=1}^{N-1} \gamma^{e,(k)}, $ for $N\geq 2$. 
This admits the incremental form
    $\gamma^{(N+1)}=\gamma^{(N)}\;+\;\frac{1}{N}\Bigl(\gamma^{e,(N)}-\gamma^{(N)}\Bigr)$, which updates the estimate by a convex combination of the previous value and the newest observation.

Denote $p_A=\P\left((F^L_{\overline{\gamma}})^{-1}(\alpha)-\gamma^a <\epsilon^{(1)}<(F_{\underline{\gamma}}^R)^{-1}(\alpha)-\gamma^a \right)$ for simpler notation. We present the following proposition showing the performance of both approaches of updating the guess. 

\begin{proposition}\label{Prop:PerformanceApproach}
    Under the Markovian update, the expected number of trials until acceptance is $\frac{1}{p_A}$. Under the CMA update, the probability of accepting the $N$-th guess is at least $1 \;-\;\frac{\sigma_\epsilon^2}{N-1}\Bigl(\tfrac1{((F^L)^{-1}(\alpha)-\gamma^a)^2}+\tfrac1{((F^R)^{-1}(\alpha)-\gamma^a)^2}\Bigr)$, which increases to 1 as $N\to\infty$. Moreover, the CMA-based strategy converges in probability to the actual strategy $S_{\gamma^a}$ at rate $1/\sqrt{N}$. 
\end{proposition}

Approach B produces error term whose dispersion around zero is strictly smaller than Approach A. By the central limit theorem, as $N$ increases, the averaged error converges in distribution to a normal law $\mathcal{N}\!\Bigl(0,\frac{\sigma_{\epsilon}^2}{N-1}\Bigr),$
highlighting that aggregating past observations not only reduces variance but also produces an approximately Gaussian perturbation. These findings confirm that the use of historical data enhances estimation stability, in line with our intuition.

In subsequent periods without client interaction, the robo-advisor may update its market outlook while the client's risk aversion remains unchanged. Because the mapping from risk preference to investment strategy depends on the market environment, a previously acceptable recommendation may then fall below the acceptance threshold.
However, assuming that the client's preferences are modeled as the PreFER process, the robo-advisor can still extract valuable information from the accepted recommendation and threshold score in previous periods. Specifically, define the acceptance probability for the guess $\gamma^{(n)}$ as a function of the actual risk aversion $\gamma^a\in(0,+\infty)$,
$G_{\gamma^{(n)}}(\gamma^a)=\frac{\pi_t^{*,a}(S_{\gamma^{(n)}})}{c_t^{(n)} \pi_t^{*,(n)}(S_{\gamma^{(n)}})}$. We then let $G_{\gamma^{(n)}}^{L}(\gamma^{a})=G_{\gamma^{(n)}}(\gamma^{a})$ for $\gamma^{a}\in(0,\gamma^{(n)})$, and $G_{\gamma^{(n)}}^{R}(\gamma^{a})=G_{\gamma^{(n)}}(\gamma^{a})$ for $\gamma^{a}\in(\gamma^{n},+\infty)$. By Lemma \ref{lemma:MonoPiPrime}, the inverse functions \((G_{\gamma^{(n)}}^L)^{-1}\) and \((G_{\gamma^{(n)}}^R)^{-1}\) exist. The client accepts the investment recommendation corresponding to \(\gamma^{(n)}\) if and only if her actual risk aversion satisfies \( (G_{\gamma^{(n)}}^L)^{-1}(\alpha) < \gamma^{a} < (G_{\gamma^{(n)}}^R)^{-1}(\alpha) \). 
In the subsequent trading period, the robo-advisor, now more confident, only needs to explore within this narrower interval. Without requiring further human involvement, it computes the acceptance interval \(\bigl[(F_{\gamma_r^a}^L)^{-1}(\alpha),(F_{\gamma_l^a}^R)^{-1}(\alpha)\bigr]\), where $\gamma_r^a=(G_{\gamma^{(n)}}^R)^{-1}(\alpha)$ and $\gamma_l^a=(G_{\gamma^{(n)}}^L)^{-1}(\alpha)$, under the new market parameters, and proposes a recommendation within this range. By Corollary \ref{Cor:AcceptableForAll}, any such recommendation attains a score above the threshold for all \( \gamma^{a} \in [G_{\gamma^{(n)}}^L)^{-1}(\alpha), (G_{\gamma^{(n)}}^R)^{-1}(\alpha)] \).

Our robo-advising framework integrates an interactive learning process that infers acceptable levels of risk aversion and portfolio choices from investor satisfaction scores, with an investment guideline, based on the PreFER process, for periods without client interaction. See Algorithm \ref{Algo} for a succinct overview.

\begin{algorithm}
\caption{PreFER-Guided Interactive Robo-Advising Algorithm}
\label{Algo}
\begin{algorithmic}
\State \textbf{Input:} $\alpha$ (acceptance threshold), $\sigma_{\epsilon}$ (feedback noise standard deviation)

\State \textbf{Output:} $\{\gamma_t^{*}\}_{t=1}^{\infty}$, $\{\pi_t^{*}\}_{t=1}^{\infty}$ (sequence of accepted risk-aversion parameters and investment strategies across periods.)
\For{$t        = 1$ to $\infty$}
    \If{the client initiates a preference update} \Comment{Interactive Period}
    \State Set the exploratory range for risk aversion $[\underline{\gamma_t},\overline{\gamma_t}]$ and the associated set of admissible strategies $[\underline{\pi}_t,\overline{\pi}_t]$.
        \State $n \leftarrow1$
        \State Recommend strategy $\pi^{\gamma^{(1)}}_t$ associated with risk aversion $\gamma_t^{(1)}\in[\underline{\gamma_t},\overline{\gamma_t}]$
        \While{$\text{client rejects} \ \pi^{\gamma^{(n)}}_t$}
            \State Receive feedback on the recommended prospect \State Infer the feedback-implied risk-aversion $\gamma_t^{e,(n)}$ through the inverse score mapping
            \State Update guess: $\gamma_t^{(n+1)} \leftarrow \gamma_t^{e,(n)}$ or $\gamma_t^{(n+1)} = \left( \prod_{k=1}^{n} \gamma_t^{e,(k)} \right)^{\frac{1}{n}}$
            \State $n \leftarrow n + 1$
            \State Recommend strategy $\pi_t^{\gamma^{(n)}}$ associated with risk aversion $\gamma_t^{(n)}$
        \EndWhile
        \State Client accepts $\pi_t^{\gamma^{(n)}}$ with score above threshold $\alpha$
        \State Set $\gamma_t^{*} \leftarrow \gamma_t^{(n)}$
    \Else 
        \Comment{Non-Interactive Period (Preference carried forward under updated market outlook)}
        \State Let $\gamma^{(\text{prev})} = \gamma^*_{t-1}$
        \State Find $\gamma_l^a$ and $\gamma_r^a$ such that 
        $G_{\gamma^{(\text{prev})}}(\gamma_l^a) = G_{\gamma^{(\text{prev})}}(\gamma_r^a) = \alpha$
        \State Update the feasible interval $[\gamma_l^a, \gamma_r^a]$ 
        \State Compute the acceptance interval \(\bigl[(F_{\gamma_r^a}^L)^{-1}(\alpha),(F_{\gamma_l^a}^R)^{-1}(\alpha)\bigr]\); select \(\gamma_t^*\) within this range
    \EndIf
    \State Execute portfolio $\pi_t^*$ using \(\gamma_t^{*}\) with current market parameters.
\EndFor
\end{algorithmic}
\end{algorithm}

\subsection{Numerical analysis}
We proceed to perform numerical analysis on investigating the behaviour of acceptance probability and the choice of acceptance boundary for non-interactive periods.
Suppose that the client accepts the recommendation corresponding to risk aversion \(\gamma_1=5\) at a threshold score \(\alpha=92\%\). Accordingly, the exploration range for the second period narrows to \([\gamma_l^a, \gamma_r^a]\),  where $\gamma_r^a=(G_{\gamma_1}^R)^{-1}(\alpha_1)$ and $\gamma_l^a=(G_{\gamma_1}^L)^{-1}(\alpha_1)$. Assume that the client's risk aversion remains unchanged, with the client specifying the exploration coefficient \(\rho_2=0.02\) for the next period. In Figure \ref{Fig:AcceptableBound&Prob}, we compute the boundaries of the acceptance interval for the next period, and the acceptance probability corresponding to guess \(\gamma^{(2)}=5.5\) if $\gamma^a=\gamma_l^a$ or $\gamma^a=\gamma_r^a$. Set \(u_1=u_2=1.15\), \(d_1=d_2=0.9\), \(p_1=0.7\), \(\underline{\gamma}=3, \overline{\gamma}=7\) and \(\rho_1=0.1\), we vary $p_2$.

As $p_2$ increases, the acceptance interval becomes narrower, so an acceptable recommendation must correspond more closely to the client's fixed risk preference. Consistently, the acceptance probability associated with an imprecise estimate of risk aversion decreases with $p_2$. The economic reason is that a higher transition probability leads to a more aggressive optimal investment strategy for a given risk-aversion parameter. Consequently, an error in the estimated risk aversion induces a larger deviation from the benchmark strategy and hence a lower acceptance score.
\begin{figure}[H]
   \includegraphics[width=7.96cm]{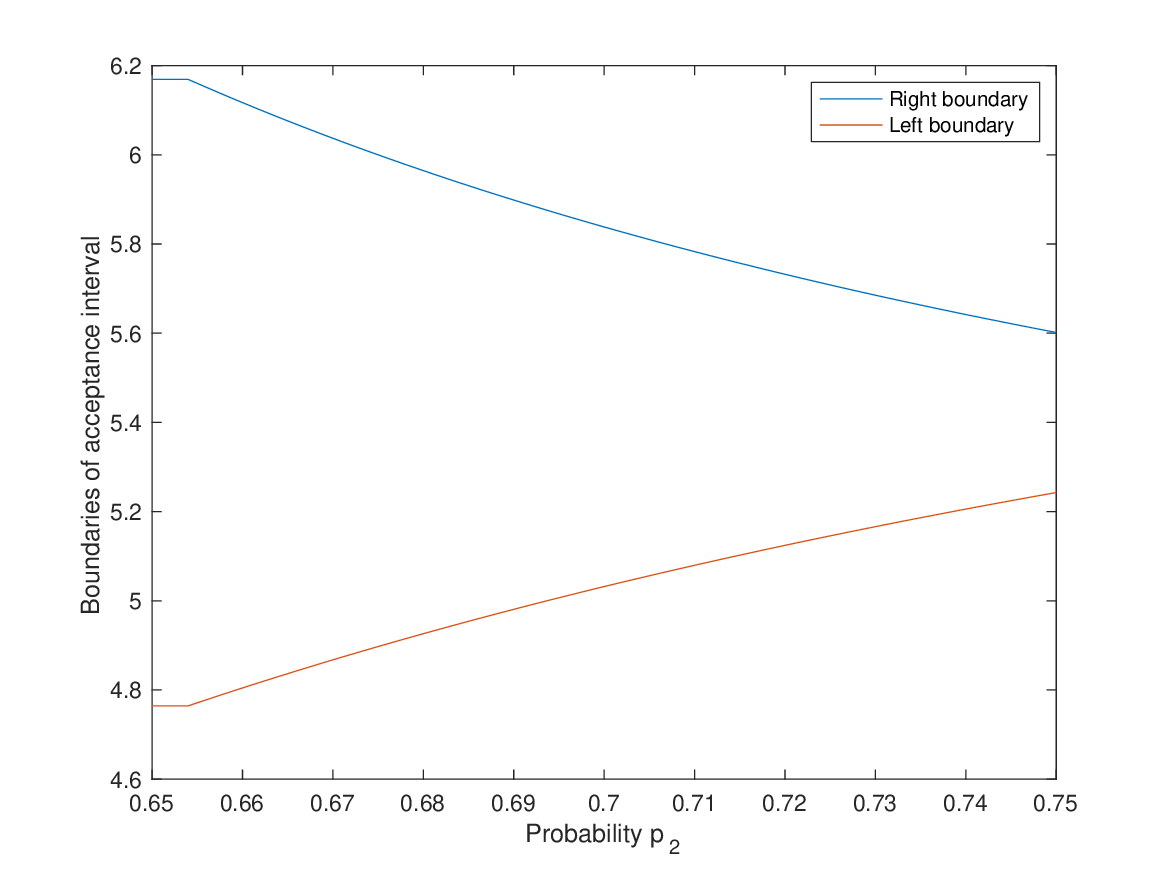}
  \includegraphics[width=7.96cm]{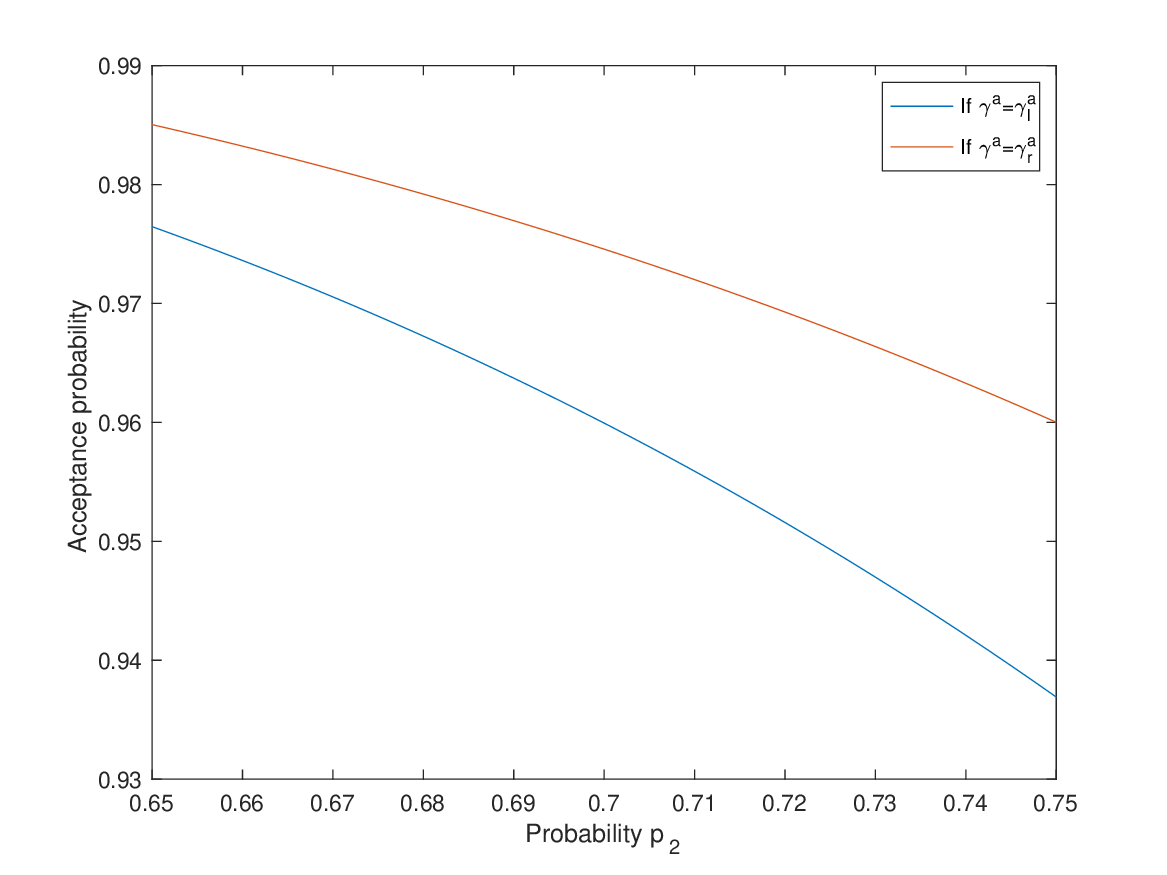} 
  \caption{acceptance interval and acceptance probability when varying transition probability}
  \label{Fig:AcceptableBound&Prob}
\end{figure}

Next, choosing $\gamma^a=5$ in Figure \ref{Fig:BenefitRobo-advising}, we investigate how the threshold score $\alpha$ influences the alignment between the accepted strategy and the actual one with full score, measured by their relative difference in the worst-case scenario, $\inf\limits_{\{\gamma^{(n)}:
\phi_{\gamma^{(n)}}\geq \alpha\}}\big\vert\frac{S_{\gamma^{(n)}}-S_{\gamma^a}}{S_{\gamma^a}}\big\vert$, where $\phi_{\gamma^{(n)}}=\frac{\pi^{*,a}(S_{\gamma^{(n)}})}{c^{(n)}\pi^{*,(n)}(S_{\gamma^{(n)}})}$. 
As the client's requirement becomes stricter with higher $\alpha$, the accepted strategy approaches the actual one even in the worst-case scenario, but this comes at the cost of requiring more guesses and iterative refinements.
The right graph shows the likelihood that, in the absence of robo-advising, the client under the PreFER process achieves the same performance as a robo-advisor operating with a specified threshold score $\alpha$. The kink points occur because the minimum acceptance probability over the exploration interval is strictly positive and is not symmetric around the actual risk aversion. 
\begin{figure}[H]
   \includegraphics[width=7.96cm]{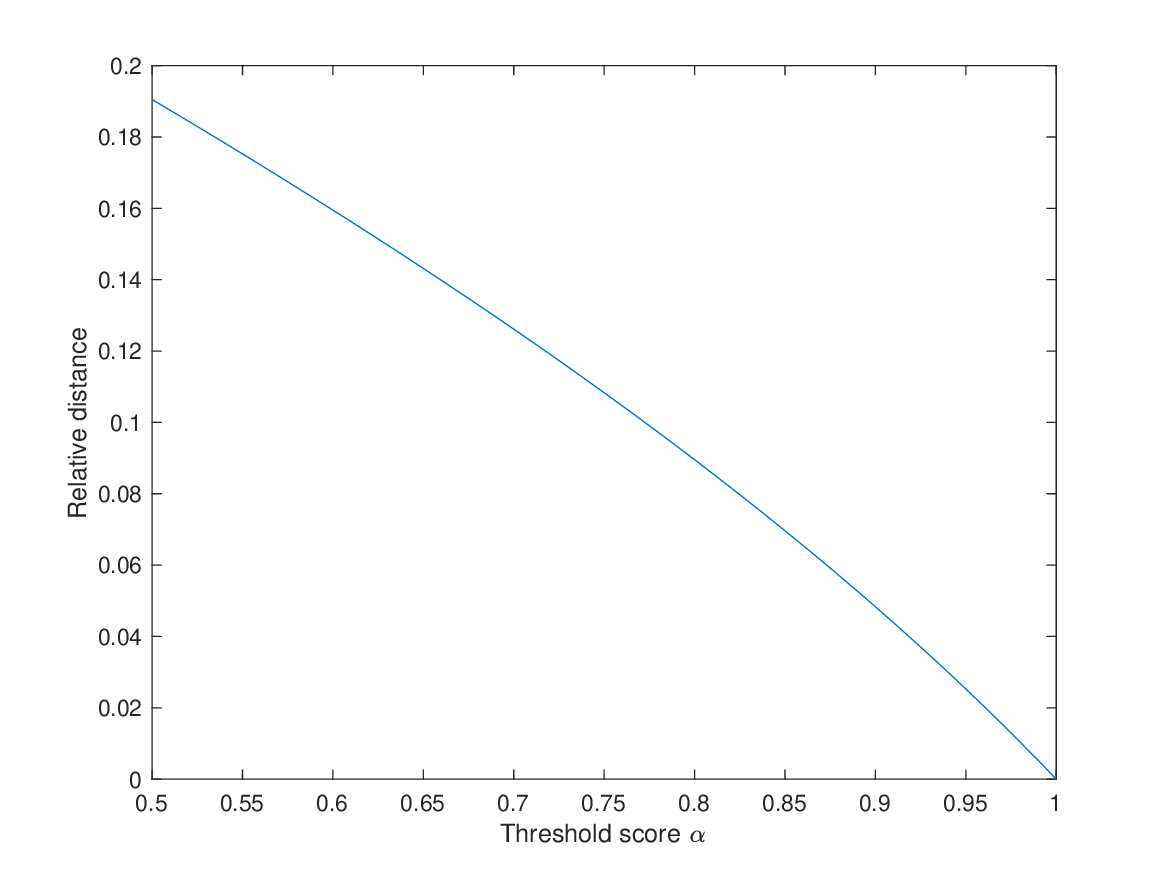}
  \includegraphics[width=7.96cm]{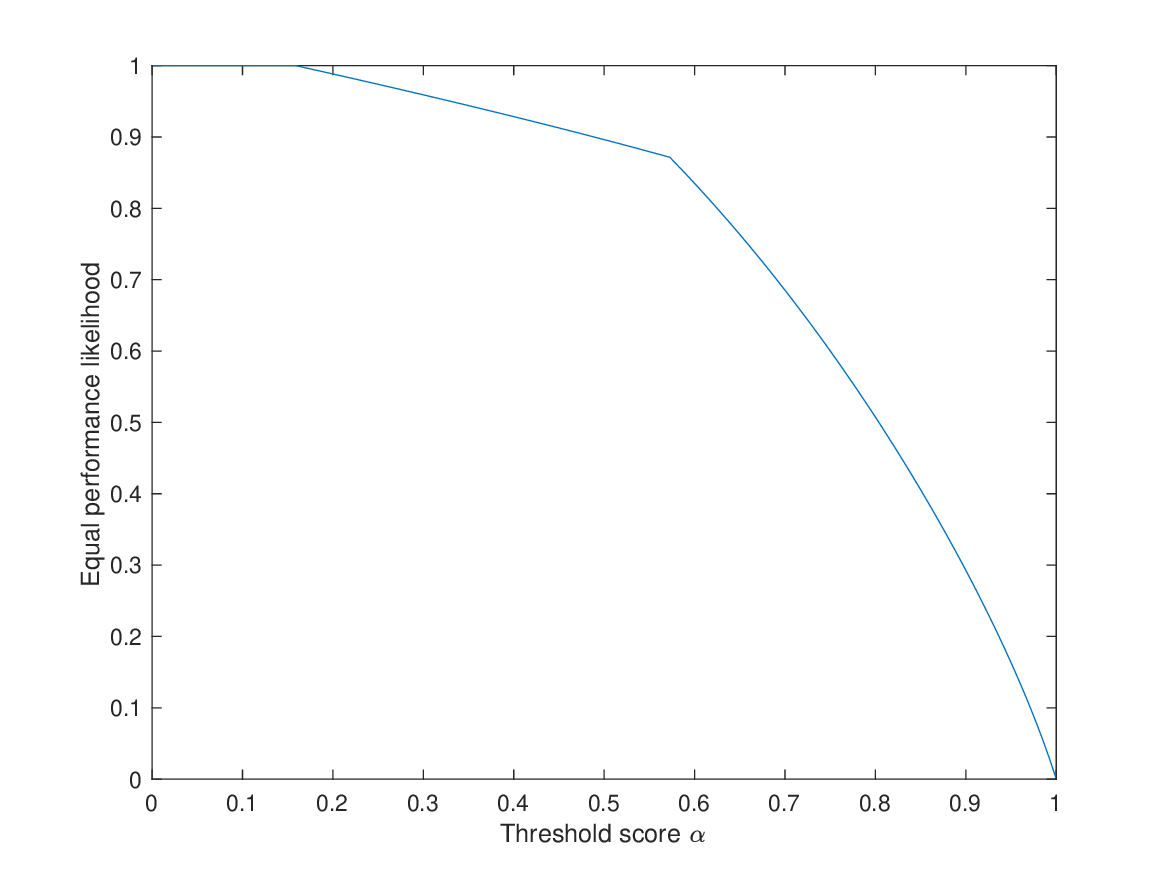} 
  \caption{Benefit of robo-advising: (a) Approaching actual strategy (b) Relative to solo investor}
  \label{Fig:BenefitRobo-advising}
\end{figure}

However, some less-engaged clients may be willing to engage in only a limited number of interactions. Figure \ref{Fig:GuaranteedScore&RequiredInteraction} examines, in the absence of a specified threshold score, what level of acceptance probability can be ensured on average when the amount of information provided by the client is limited.
Let the error $\epsilon^{(n)}$ follow a uniform distribution $\mathcal{U}(-d_{\epsilon}, d_{\epsilon})$, we evaluate the performance under Approach B via Monte Carlo simulation, averaged over $5 \times 10^6$ trials.
As the level of feedback inconsistency, measured by $d_\epsilon$, increases, the client's responses become less precise, reducing the effectiveness of the interactive scheme and resulting in a lower minimum guaranteed score.
In addition, as the client provides more feedback with increasing $n$, the system can ensure a higher score, though the marginal improvement from each additional interaction gradually diminishes.
We further compute in Figure \ref{Fig:GuaranteedScore&RequiredInteraction} the expected number of inferences required for making a satisfactory recommendation for a client whose feedback noise $d_\epsilon$ is affected by idiosyncratic shocks.
Consider a specific type of idiosyncratic life-event shock, such as a health-related or emotional shock, that may affect the client at interactions, occurring independently at each interaction with probability $p^{\text{shock}}$. Set inconsistency level $d_{\epsilon}=1$ which under shock becomes $d_{\epsilon}^{\text{shock}}=1.9$. We observe that when the client is more easily influenced or sets a higher satisfaction threshold, more interactions are required for inference.
\begin{figure}[H]
   \includegraphics[width=7.96cm]{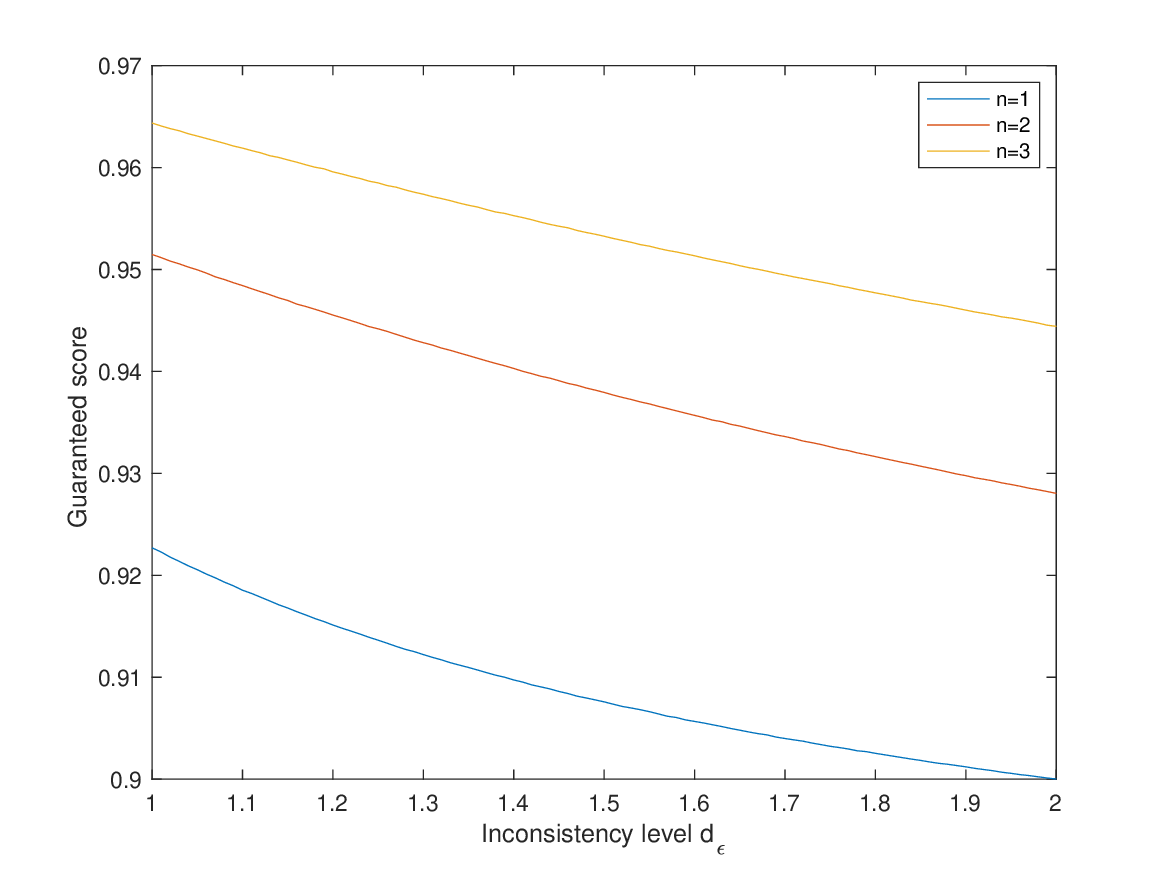}
  \includegraphics[width=7.96cm]{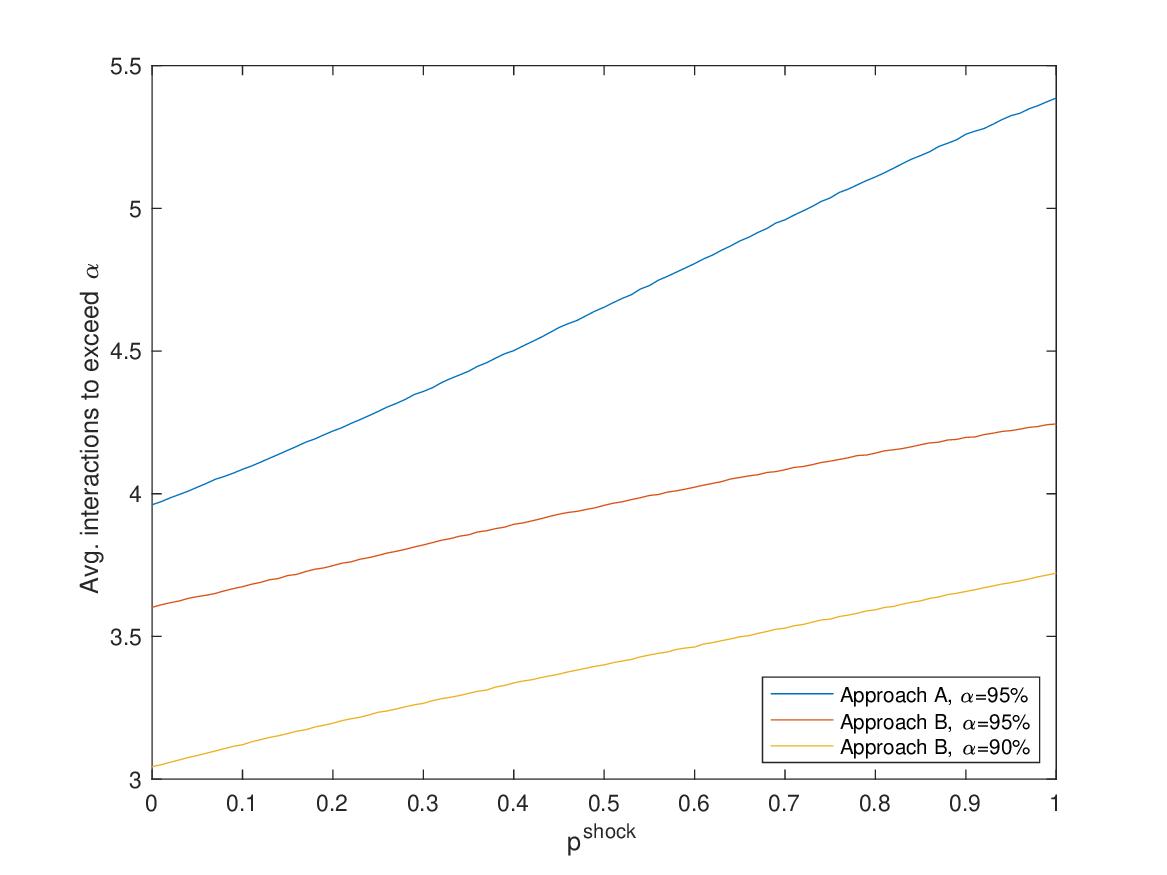} 
  \caption{Performance guarantee and interaction effort for target satisfaction}
  \label{Fig:GuaranteedScore&RequiredInteraction}
\end{figure}


\begin{table}[ht]
\centering
\begin{tabular}{lccccccccc}
\toprule
\textbf{Path} & Int. 1 & Int. 2 & Int. 3 & Int. 4 & Int. 5 & Int. 6 & Int. 7 & Int. 8 & Int. 9 \\
\midrule
Seed 1, A & 87.7\% & \textbf{\textcolor{red}{66.6\%}} & \textbf{\textcolor{red}{82.4\%}} & \textbf{\textcolor{red}{47.8\%}} & 91.4\% & 92.5\% & 91.5\% & 87.9\% & 90.5\% \\
Seed 1, B & 87.7\% & \textbf{\textcolor{red}{66.6\%}} & \textbf{\textcolor{red}{75.4\%}} & \textbf{\textcolor{red}{94.1\%}} & 93.5\% & 95.6\% & 97.0\% & 95.9\% & 95.4\% \\
Seed 2, A & 87.7\% & \textbf{\textcolor{red}{69.3\%}} & 92.3\% & 99.4\% & \textbf{\textcolor{red}{43.1\%}} & \textbf{\textcolor{red}{98.8\%}} & 96.9\% & 58.7\% & 67.1\% \\
Seed 2, B & 87.7\% & \textbf{\textcolor{red}{69.3\%}} & 83.2\% & 90.6\% & \textbf{\textcolor{red}{97.6\%}} & \textbf{\textcolor{red}{97.9\%}} & 97.7\% & 99.3\% & 99.7\% \\
\bottomrule
\end{tabular}
\caption{Simulated client scores affected by random shocks}
\label{tab:simulated_scores_AB}
\end{table}

Table \ref{tab:simulated_scores_AB} presents four simulated score paths across interactions, two under Approach A and two under Approach B, grouped by shared random seeds. All simulations use $p^{\text{\text{shock}}} = 0.4$, and scores are rounded to one decimal place. Red entries indicate interactions in which an event shock occurs. 
Although Approach A may occasionally achieve a higher score earlier than Approach B, it is more vulnerable to sudden drops caused by random shocks. In contrast, Approach B shows a steadier upward trend in satisfaction, as the inconsistency error gradually diminishes when using a cumulative moving average.


\section{Conclusion}\label{sec:conclusion}
We develop an interactive scoring-based preference-learning framework, based on discrete-time Predictable Forward Exploratory Reward (PreFER) processes, to infer the client's risk aversion. By modeling the client's indecisiveness and investment preferences through relaxed stochastic control, we propose an interactive learning procedure in which the client's evaluation of a recommended investment prospect is modeled through an acceptance probability, which serves as a latent satisfaction score. We then analyze the properties of this acceptance probability to justify the scoring scheme and to show how the robo-advisor uses the observed feedback to update its estimate of the client's risk aversion. Our results show that PreFER processes provide a tractable framework for learning client preferences through interaction and carrying the learned preference forward to generate personalized recommendations under evolving market conditions.
At each new time point, conditional on the available information, we identify an acceptance interval from which any recommended strategy achieves a score above the prescribed threshold, even as market conditions evolve. This enables automated updating of the investment recommendation with minimal human intervention.


An important direction for future research is to study the effect of market-model misspecification on the implementation of learned preferences. In particular, while preference learning in the present framework is conducted conditional on the market outlook communicated by the robo-advisor, a robust extension could account for uncertainty in translating the learned preference into investment recommendations. Another natural direction is to extend the inverse-learning analysis beyond the scalar CARA specification to richer or multidimensional preference classes, for which one-to-one identification may no longer hold.

\section*{Acknowledgments}

The authors are grateful to Bahman Angoshtari, Sylvain Carré, Moris Strub, Gechun Liang, Seyoung Park, Epaminondas Kyriakidis, and Yong Hyun Shin for their helpful comments and suggestions.
This work was presented at the 13th Bachelier World Congress, the 34th European Conference on Operational Research, the Stochastic Finance at Warwick Seminars, and the 9th Asian Quantitative Finance Conference.
The authors thank the participants for their comments.

\section*{Appendix A: Multiple risky assets}\label{subsec:MultiStock}

Suppose that there are \(N\) binomial stocks in the market, the return $R_t^{(i)}$ of each stock can increase to the price level \(u_t^{(i)}\) with probability \(p_t^{(i)}\) or decrease to \(d_t^{(i)}\) with probability \(1 - p_t^{(i)}\). The investor chooses the amount \(f_t^{(i)}\) to invest in the stock \(i\), and the rest in the bond for each trading period $[t-1,t)$. Then the dynamics of the wealth process becomes $X_t=X_{t-1}+\sum\limits_{i=1}^N f_t^{(i)}(R_t^{(i)}-1)$.

In this more general setting, we consider relaxed control as a distributional control process, characterized by its density function \(\pi^{N} = \pi_t(s_t^{(1)},\dots,s_t^{(N)}), t \in \mathbb{N}\), which represents the joint probability density of investment across all stocks. Accordingly, the expected performance of a strategy is computed by integrating against the sequence of density functions involved, which after bundling the nested integrals and density products, is then 
$$
\mathbb{E}_{\pi^N}\bigl[U(X_t)\,\big\vert\,\mathcal{F}_{t}\bigr]
\;=\;
\int_{\Pi_1}\!\!\int_{\Pi_2}\!\cdots\int_{\Pi_t}
U\Bigl(x \;+\;\sum_{i=1}^t\sum_{k=1}^N s_i^k\bigl(R_i^k-1\bigr)\Bigr)
\;\prod_{i=1}^t \pi_i^N(s_i)\;ds_1\,ds_2\cdots ds_t,
$$
where $
\Pi_i \;=\; \prod_{k=1}^N \bigl[\underline{\pi}_i^{(k)},\,\overline{\pi}_i^{(k)}\bigr], 
s_i = (s_i^1,\dots,s_i^N), 
$
$
\pi_i^N(s_i)\;=\;\pi_i^N\bigl(s_i^1,\dots,s_i^N\bigr),
$
and
$ 
ds_i = ds_i^1\,ds_i^2\cdots ds_i^N.
$
The term inside the operator \(\esssup\) of equation \eqref{eq:InverseEqn_Dynamic} becomes
\begin{align*}
    \E_{\pi^*} \left[ U_t \left( X_{t-1} + \sum\limits_{i=1}^N f_t^{(i),*}(R_t^{(i)}-1)\right)-\lambda_t \int_{\Pi_t}(\pi_t^*)^2(s_t)ds_t\bigg\vert \mathcal{F}_{t-1}\right].
\end{align*}

To explicitly model the correlation structure among different stocks, we assume that the probability space $(\Omega,\mathcal{F},\P)$ supports another stochastic processes $(\xi_t^{cn})_{t\in \mathbb{N}}$ which is not tradeable by any investor. 
Each $\xi_t^{cn}$ represents a non-traded stochastic factor for period $[t-1,t)$ and follows a Bernoulli distribution with transition probability $p_t^{cn}$, $(p_t^{cn})_{t\in \mathbb{N}}$ are allowed to be stochastic. 
The event $\{ \xi_t^{cn} = 1 \}$ represents a bull market regime and $\{ \xi_t^{cn} = 0 \}$ represents a bear market regime for period $[t-1,t)$. 

We define the filtration $\mathbb{F} = (\mathcal{F}_t)_{t \in \mathbb{N}_0}$ by letting $\mathcal{F}_t$ be the augmented $\sigma$-algebra generated by $(\xi_j^{cn})_{j=1}^t$,  $(p_j^{cn})_{j=1}^{t+1}$, $(B^{(i)}_j)_{j=1}^t$, and the market parameters $(u^{(i)}_j, d^{(i)}_j, p^{(i)}_j)_{j=1}^{t+1}$ of all the stocks.
Denote the transition probability of the stochastic factor by $p^{cn}_t:=\P[\xi_t^{cn}=1\vert \mathcal{F}_{t-1}]$. The outcome of the non-traded stochastic factor and stock returns are only known at the end of each trading period. 
Let $\sigma(\xi_t^{cn})$ denote the $\sigma$-algebra generated by the random variable $\xi_t^{cn}$.
We introduce the conditional probabilities for an upward move of the $i$th stock in the bull or bear market, defined by $p_t^{(i),1}=\frac{\P[\xi_t^{cn}=1, B_t^{(i)}=1\vert \mathcal{F}_{t-1}]}{\P[\xi_t^{cn}=1\vert \mathcal{F}_{t-1}]}$ and $p_t^{(i),0}=\frac{\P[\xi_t^{cn}=0, B_t^{(i)}=1\vert \mathcal{F}_{t-1}]}{\P[\xi_t^{cn}=0\vert \mathcal{F}_{t-1}]}$. 
One can easily deduce that 
$$\P[B_t^{(i)}=1\big\vert \mathcal{F}_{t-1}\vee \sigma(\xi_t^{cn})]=p_t^{(i),1}\mathbbm{1}_{\{\xi_t^{cn}=1\}}+p_t^{(i),0}\mathbbm{1}_{\{\xi_t^{cn}=0\}}.$$

We suppose that $p_t^{(i),1}>p_t^{(i),0}$ almost surely for $i = 1, \dots,N$, that is, all stocks are more likely to perform well in a bull market than in a bear market regime. 
While probabilities of an upward move of stocks depend on the non-traded stochastic factor, price levels are assumed to be independent thereof. 
We further assume that $B^{(i)}_t$ and $B^{(\ell)}_t$ are conditionally independent given $\mathcal{F}_{t-1}\vee\sigma(\xi^{cn}_t)$ for any $i \neq \ell$.
That is, the dependence structure between different stocks is completely characterized by the non-traded stochastic factor.

Next, we focus on the case of \( N = 2 \) and CARA risk preferences for simplicity and as an illustration. We then solve the equation
\begin{align*}
U_{t-1}(X_{t-1})=\sup\limits_{\pi_t\in\mathcal{A}(t)}&\int_{\underline{\pi}_t^{(1)}}^{\overline{\pi}_t^{(1)}}\int_{\underline{\pi}_t^{(2)}}^{\overline{\pi}_t^{(2)}}\E\left[U_t\left(X_{t-1}+\sum\limits_{i=1}^2 s^{(i)}(R_t^{(i)}-1)\right)\Bigg\vert \mathcal{F}_{t-1}\right]\pi_t(s^{(1)},s^{(2)})ds^{(1)}ds^{(2)}\\& +\lambda_t H(\pi_t).
\end{align*}

Consider $\rho_t>0$ that is $\mathcal{F}_t$-measurable, define
$P_t=C_{t-1}J_t$ and $C_t=\prod\limits_{n=1}^{t}P_n$ with $C_0=1$ iteratively, where $J_t=\frac{1}{\beta_t-\frac{1}{4\rho_t}\int_{\underline{\pi}_t^{(1)}}^{\overline{\pi}_t^{(1)}}\int_{\underline{\pi}_t^{(2)}}^{\overline{\pi}_t^{(2)}}h_t^2(s^{(1)},s^{(2)})ds^{(1)}ds^{(2)}},$ and
\begin{equation}\label{Expression-h-N=2}
    \begin{aligned}
    h_t(s^{(1)},s^{(2)})=-C_{t-1}\Bigg(&p_t^{cn}\prod\limits_{i=1}^2\left(p_t^{(i),1}e^{-\gamma s^{(i)}(u_t^{(i)}-1)}+(1-p_t^{(i),1})e^{-\gamma s^{(i)}(d_t^{(i)}-1)}\right)\\&\ +(1-p_t^{cn})\prod\limits_{i=1}^2\left(p_t^{(i),0}e^{-\gamma s^{(i)}(u_t^{(i)}-1)}+(1-p_t^{(i),0})e^{-\gamma s^{(i)}(d_t^{(i)}-1)}\right)\Bigg)+\beta_t.
\end{aligned}
\end{equation}
with 
\begin{align*}
    \beta_t=\frac{2\rho_t-\int_{\underline{\pi}_t^{(1)}}^{\overline{\pi}_t^{(1)}} \int_{\underline{\pi}_t^{(2)}}^{\overline{\pi}_t^{(2)}}\left(h_t(s^{(1)},s^{(2)})-\beta_t\right)ds^{(1)}ds^{(2)}}{
(\overline{\pi}_t^{(1)}-\underline{\pi}_t^{(1)})(\overline{\pi}_t^{(2)}-\underline{\pi}_t^{(2)})}.
\end{align*}
We consider temperature parameter of the form $\lambda_t:=\lambda(X_t)=\rho_t P_te^{-\gamma X_{t-1}}$ for control problem over trading period $[t-1,t)$,

Again, to ensure the well-posedness of the PreFER processes, we assume that for any prescribed exploration coefficient $\rho_t$, the intervals \([ \underline{\pi}_t^{(1)}, \overline{\pi}_t^{(1)} ]\) and \([ \underline{\pi}_t^{(2)}, \overline{\pi}_t^{(2)} ]\) must be chosen such that the following conditions are satisfied: 
\begin{align}\label{AssumptionForPositive-N=2}
    \beta_t-\frac{1}{4\rho_t}\int_{\underline{\pi}_t^{(1)}}^{\overline{\pi}_t^{(1)}}\int_{\underline{\pi}_t^{(2)}}^{\overline{\pi}_t^{(2)}}h_t^2(s^{(1)},s^{(2)})ds^{(1)}ds^{(2)} > 0, \ h_t(\overline{\pi}_t^{(1)},\overline{\pi}_t^{(2)}) > 0 \  \text{and} \  h_t(\underline{\pi}_t^{(1)},\underline{\pi}_t^{(2)}) > 0.
\end{align}

The following theorem provides an explicit construction of the PreFER process for the case of two binomial stocks.
\begin{theorem}\label{Thm:ExploForwardRL-N=2}
For $t\in \mathbb{N}$, the process $U_t(x)=-C_te^{-\gamma x}$, is a PreFER process.
The density function of the optimal strategy $f_t^{*}$ for $t\in \mathbb{N}$ is given by
\begin{align}\label{OptimalStrategyRL-N=2}
    \pi_t^{*}(s^{(1)},s^{(2)})=\frac{h_t(s^{(1)},s^{(2)})}{2\rho_t},
\end{align}
which generates the corresponding optimal wealth process
\begin{align}\label{OptimalWealthRL-N=2}
     X_t^{*}= X_0+\sum_{n=1}^{t}\left(f_n^{(1),*}(R_n^{(1)}-1)+f_n^{(2),*}(R_n^{(2)}-1)\right).
\end{align}
\end{theorem}

\section*{Appendix B: CRRA risk preferences}
In this example, we focus on CRRA utilities, the utility function is given by $U_0(x)=(1-\frac{1}{\theta})^{-1}x^{1-\frac{1}{\theta}}, x \in \mathbb{R}_+$, where the constant parameter $\theta$ represents the risk tolerance. As is common in power utility models, the investment policy $f_{t}$ is taken to be a fraction of wealth, the wealth process then evolves according to $X_t=X_{t-1}(1+f_{t}(R_{t}-1))$.

Consider $\rho_t>0$ that is $\mathcal{F}_t$-measurable. To avoid introducing additional notation, we define analogously in this example
$P_t=C_{t-1}J_t$ and $C_t=\prod\limits_{n=1}^{t}P_n$ with $C_0=1$ iteratively, where 
$J_t=\frac{1}{\frac{1}{4\rho_t}\int_{\underline{\pi}_t}^{\overline{\pi}_t}h_t^2(s)ds-\beta_t},$
\begin{align}\label{Expression-h-CRRA}
    h_t(s)=C_{t-1}p_t\left(1+s(u_t-1)\right)^{1-\frac{1}{\theta}}-C_{t-1}(1-p_t)\left(1+s(d_t-1)\right)^{1-\frac{1}{\theta}}+\beta_t,
\end{align}
and 
\begin{align}\label{Expression-beta-CRRA}
    \beta_t=\frac{2\rho_t-C_{t-1}\left(\frac{p_t\left(1+\overline{\pi}_t(u_t-1)\right)^{2-\frac{1}{\theta}}}{(2-\frac{1}{\theta})(u_t-1)}-\frac{p_t\left(1+\underline{\pi}_t(u_t-1)\right)^{2-\frac{1}{\theta}}}{(2-\frac{1}{\theta})(u_t-1)}+\frac{(1-p_t)\left(1+\overline{\pi}_t(d_t-1)\right)^{2-\frac{1}{\theta}}}{(2-\frac{1}{\theta})(d_t-1)}-\frac{(1-p_t)\left(1+\underline{\pi}_t(d_t-1)\right)^{2-\frac{1}{\theta}}}{(2-\frac{1}{\theta})(d_t-1)}\right)}{\overline{\pi}_t-\underline{\pi}_t}.
\end{align}
We consider temperature parameter of the form $\lambda_t:=\lambda(X_t)=\rho_t P_t(1-\frac{1}{\theta})X_{t-1}^{1-\frac{1}{\theta}}$ for control problem over trading period $[t-1,t)$.

To ensure the well-posedness of the PreFER processes under CRRA risk preferences, we assume that for any prescribed exploration coefficient \(\rho_t\), the interval \([ \underline{\pi}_t, \overline{\pi}_t ]\) must be chosen such that the following conditions are satisfied: 
\begin{align}\label{AssumptionForPositive-CRRA}
    \frac{1}{4\rho_t} \int_{\underline{\pi}_t}^{\overline{\pi}_t} h_t^2(s) \, ds - \beta_t> 0, \quad h_t(\overline{\pi}_t) > 0 \quad \text{and} \quad h_t(\underline{\pi}_t) > 0.
\end{align}

We next present the main theorem for constructing PreFER process in the case of CRRA risk preferences.
\begin{theorem}\label{Thm:ExploForwardRL-CRRA}
For $t\in \mathbb{N}$, the process $U_t(x)=C_t(1-\frac{1}{\theta})^{-1}x^{1-\frac{1}{\theta}}$, is a PreFER process.
The density function of the optimal strategy $f_t^{*}$ for $t\in \mathbb{N}$ is given by
\begin{align}\label{OptimalStrategyRL-CRRA}
    \pi_t^{*}(s)=\frac{C_{t-1}p_t\left(1+s(u_t-1)\right)^{1-\frac{1}{\theta}}+C_{t-1}(1-p_t)\left(1+s(d_t-1)\right)^{1-\frac{1}{\theta}}+\beta_t}{2\rho_t},
\end{align}
which generates the corresponding optimal wealth process
\begin{align}\label{OptimalWealthRL-CRRA}
     X_t^{*}= X_0+\sum_{n=1}^{t}f_n^{*}(R_n-1).
\end{align}
\end{theorem}

Again, \(\pi^{*}_t(s)\) reaches its maximum at the benchmark strategy obtained by setting \(\rho = 0\), corresponding to the non-exploratory optimal strategy in the classical portfolio selection problems. Indeed, the derivative of $\pi^{*}_t(s)$ over $s$ is computed by
$$(\pi_t^{*})^{'}(s)=\frac{C_{t-1}}{2\rho_t}(1-\frac{1}{\theta})
\left(p_t(u_t-1)\left(1+s(u_t-1)\right)^{-\frac{1}{\theta}}+(1-p_t)(d_t-1)\left(1+s(d_t-1)\right)^{-\frac{1}{\theta}}\right),$$
letting $(\pi^{*}_t)^{'}(s^*)=0$ we obtain
$s^*=\frac{1-\left(\frac{p_t(1-q_t)}{q_t(1-p_t)}\right)^{\theta}}{\left(\frac{p_t(1-q_t)}{q_t(1-p_t)}\right)^{\theta}(d_t-1)-(u_t-1)}=\frac{1}{u_t-1}\left(\frac{p_t^{\theta}q_t^{-\theta}}{p_t^{\theta}q_t^{1-\theta}+(1-p_t)^{\theta}(1-q_t)^{1-\theta}}-1\right)$,
which equals to the optimal strategy when there is no exploration.

\section*{Appendix. Proofs}
\renewcommand{\thesubsection}{\Alph{subsection}}

\subsection{Proof of Proposition \ref{Prop:MaxControlDist}}
    In order to solve for the optimal distributional control that satisfies \eqref{PreFEROptimalControl}, we define
    the Lagrangian function:
\begin{align*}
    L(\pi_t,\beta_t)=\int_{\underline{\pi}_t}^{\overline{\pi}_t}\left(\E \left[ U_t \left( X_{t-1} + s(R_t-1)\right)\vert \mathcal{F}_{t-1}\right]-\lambda_t \pi_t(s)\right)\pi_t(s)ds+\beta_t \left(\int_{\underline{\pi}_t}^{\overline{\pi}_t}\pi_t(s)ds-1\right),
\end{align*} 
then by applying the method of Lagrange multipliers, we obtain the following system of equations:
$$\int_{\underline{\pi}_t}^{\overline{\pi}_t}\left(\E \left[ U_t \left( X_{t-1} + s(R_t-1)\right)\vert \mathcal{F}_{t-1}\right]-2\lambda_t \pi_t(s)+ \beta_t \right)ds =0,$$
$$\int_{\underline{\pi}_t}^{\overline{\pi}_t}\pi_t(s)ds - 1=0.$$
Applying the usual verification technique, we can solve the equations of the constrained optimization problem and obtain an optimal control whose density function is given by
$$\pi_t^{*}(s)=\frac{E \left[ U_t \left( X_{t-1} + s(R_t-1)\right)\vert \mathcal{F}_{t-1}\right]+\beta_t}{2\lambda_t},$$ with $\beta_t=\frac{2\lambda_t-\int_{\underline{\pi}_t}^{\overline{\pi}_t}E \left[ U_t \left( X_{t-1} + s(R_t-1)\right)\vert \mathcal{F}_{t-1}\right]ds}{\overline{\pi}_t-\underline{\pi}_t}.$

Note that maximizing \(\pi_t^{*}(s)\) is equivalent to maximizing \( E \left[ U_t \left( X_{t-1} + s(R_t-1)\right) \mid \mathcal{F}_{t-1} \right] \), which corresponds to the optimization problem in Equation \eqref{PFPPOptimalControl}. This establishes the claim.
\qed

\subsection{Proof of Theorem \ref{Thm:ExploForwardRL}}
We start with the following auxiliary result for constructing the PreFER processes within the CARA family.
\begin{lemma}\label{LemmaSup}
Suppose that \( f_t \) is the random control process associated with the probability density \(\pi_t\). For any $t\in\mathbb{N}$, we have
\begin{align}\label{SupremumEqn}
    \sup_{\pi_t\in \mathcal{A}(t)}\E_{\pi}\left[-C_{t-1}e^{-\gamma f_t(R_t-1)}-\rho_t\int_{\underline{\pi}_t}^{\overline{\pi}_t}\pi_t^2(s)ds\Bigg\vert \mathcal{F}_{t-1}\right]=\frac{1}{4\rho_t}\int_{\underline{\pi}_t}^{\overline{\pi}_t}h_t^2(s)ds-\beta_t.
\end{align}
\end{lemma}
\begin{proof}
    Since $\E_{\pi}\left[-C_{t-1}e^{-\gamma f_t(R_t-1)}\big\vert \mathcal{F}_{t-1}\right]=-C_{t-1}
    \int_{\underline{\pi}_t}^{\overline{\pi}_t}
    \left(p_te^{-\gamma s(u_t-1)}+(1-p_t)e^{-\gamma s(d_t-1)}\right)\pi_t(s)ds$, we set the Lagrangian function
\begin{align*}
    L(\pi_t,\beta_t)=-&\int_{\underline{\pi}_t}^{\overline{\pi}_t}\left(C_{t-1}\left(p_te^{-\gamma(u_1-1)s}+(1-p_t)e^{-\gamma(d_1-1)s}\right)+\rho_t \pi_t(s)\right)\pi_t(s)ds\\& \qquad+\beta_t \left(\int_{\underline{\pi}_t}^{\overline{\pi}_t}\pi_t(s)ds-1\right),
\end{align*} 
then by the method of Lagrange multipliers, we obtain the following system of equations:
$$\int_{\underline{\pi}_t}^{\overline{\pi}_t}\left(-C_{t-1}p_te^{-\gamma(u_t-1)s}-C_{t-1}(1-p_t)e^{-\gamma(d_t-1)s}-2\rho_t \pi_t(s)+ \beta_t \right)ds =0,$$
$$\int_{\underline{\pi}_t}^{\overline{\pi}_t}\pi_t(s)ds - 1=0.$$
Applying the usual verification technique, we can solve the equations of the constrained optimization problem and obtain an optimal control whose density function is given by
$$\pi_t^{*}(s)=\frac{-C_{t-1}p_te^{-\gamma(u_t-1)s}-C_{t-1}(1-p_t)e^{-\gamma(d_t-1)s}+\beta_t}{2\rho_t},$$ with $\beta_t=\frac{2\rho_t+C_{t-1}\int_{\underline{\pi}_t}^{\overline{\pi}_t}\left(p_te^{-\gamma(u_t-1)s}+(1-p_t)e^{-\gamma(d_t-1)s}\right)ds}{\overline{\pi}_t-\underline{\pi}_t},$ which is further equal to \eqref{Expression-beta} after direct computation.
Additionally, to satisfy the requirements of a probability density function, it is crucial that \(\pi_t(s) \geq 0\). This condition is guaranteed by Assumption \eqref{AssumptionForPositive}.

Next, by plugging $\pi_t^{*}(s)$ back in the LHS of \eqref{SupremumEqn}, it becomes
\begin{align*}
&\frac{1}{4\rho_t}\int_{\underline{\pi}_t}^{\overline{\pi}_t}(-C_{t-1}p_te^{-\gamma(u_t-1)s}-C_{t-1}(1-p_t)e^{-\gamma(d_t-1)s}+\beta_t)^2ds\\ & -\frac{\beta_t}{2\rho_t}\int_{\underline{\pi}_t}^{\overline{\pi}_t}(-C_{t-1}p_te^{-\gamma(u_t-1)s}-C_{t-1}(1-p_t)e^{-\gamma(d_t-1)s}+\beta_t)ds\\=&\frac{1}{4\rho_t}\int_{\underline{\pi}_t}^{\overline{\pi}_t}(-C_{t-1}p_te^{-\gamma(u_t-1)s}-C_{t-1}(1-p_t)e^{-\gamma(d_t-1)s}+\beta_t)^2ds-\beta_t,
\end{align*}
which proves the claim, by letting $h_t(s)$ be given by (\ref{Expression-h}).
\end{proof}

We continue with the proof of Theorem \ref{Thm:ExploForwardRL}.

    Conditions $(i)$ and $(ii)$ of Definition \ref{def:ExploratoryFrowardPreferences} follow directly.
Indeed, let $x \in \mathbb{R}$, we first note that for any $t\in \mathbb{N}$ and $s\in[\underline{\pi}_t,\overline{\pi}_t]$, the sequence of $(J_n)_{n=1}^t$ are Borel-measurable functions of the market parameters $(p_n)_{n=1}^t, (u_n)_{n=1}^t,(d_n)_{n=1}^t, (\rho_n)_{n=1}^t$ which are $\mathcal{F}_{t-1}$-measurable. 
It then follows that the process $U_t(x)=-C_te^{-\gamma x}$ is $\mathcal{F}_{t-1}$-measurable.

Next, we show Condition $(iii)$ of Definition \ref{def:ExploratoryFrowardPreferences}. 
Let $t\in \mathbb{N}$ and consider any arbitrary ${X}\in {\mathcal{X}}(x)$, we must show that
\begin{align*}
    \E_{\pi}\left[U_{t}(X_{t})\vert \mathcal{F}_{t}\right]\geq \E_{\pi}\left[U_{t+1}(X_{t+1})-\rho_{t+1} P_{t+1}e^{-\gamma X_{t}}\int_{\underline{\pi}_{t+1}}^{\overline{\pi}_{t+1}}\pi_{t+1}^2(s)ds\bigg\vert \mathcal{F}_{t}\right].
\end{align*}

Straightforward computation shows
\begin{align*}
    &\E_{\pi}\left[U_{t+1}(X_{t+1})-\rho_{t+1} P_{t+1}e^{-\gamma X_{t}}\int_{\underline{\pi}_t}^{\overline{\pi}_t}\pi_{t+1}^2(s)ds\bigg\vert \mathcal{F}_{t}\right]\\=&P_{t+1}
    \E_{\pi}\left[e^{-\gamma X_{t}}\vert \mathcal{F}_{t}\right]
    \E_{\pi}\left[-C_te^{-\gamma f_{t+1}(R_{t+1}-1)}-\rho_{t+1} \int_{\underline{\pi}_{t+1}}^{\overline{\pi}_{t+1}}\pi_{t+1}^2(s)ds\bigg\vert \mathcal{F}_{t}\right].
\end{align*}

Furthermore, $U_t(X_t)=-C_te^{-\gamma X_{t}}$, and $P_{t+1}=C_tJ_{t+1}$. Hence, it remains to show that
\begin{align*}
    -1\geq J_{t+1}\E_{\pi}\left[-C_te^{-\gamma f_{t+1}(R_{t+1}-1)}-\rho_{t+1} \int_{\underline{\pi}_{t+1}}^{\overline{\pi}_{t+1}}\pi_{t+1}^2(s)ds\bigg\vert \mathcal{F}_{t}\right],
\end{align*}
which easily follows using Lemma \ref{LemmaSup}.

To show Condition $(iv)$, we let $X_t^{*}$, $t \in \mathbb{N}_0$, be given by (\ref{OptimalWealthRL}). We need to establish
\begin{align*}
    \E_{\pi}\left[U_{t}(X_{t}^*)\vert \mathcal{F}_{t}\right]= \E_{\pi}\left[U_{t+1}(X_{t+1}^{*})-\rho_{t+1} P_{t+1}e^{-\gamma X_{t}^*}\int_{\underline{\pi}_{t+1}}^{\overline{\pi}_{t+1}}(\pi_{t+1}^*)^2(s)ds\bigg\vert \mathcal{F}_{t}\right].
\end{align*}
Since ${X}_{t+1}^{*}={X}_{t}^{*}+f_{t+1}^{*}(R_{t+1}-1)$, the above equality reduces to 
\begin{align*}
    -1= J_{t+1}\E_{\pi}\left[-C_te^{-\gamma f_{t+1}^*(R_{t+1}-1)}-\rho_{t+1} \int_{\underline{\pi}_{t+1}}^{\overline{\pi}_{t+1}}(\pi_{t+1}^*)^2(s)ds\bigg\vert \mathcal{F}_{t}\right].
\end{align*}
This then directly follows with Lemma \ref{LemmaSup}.
\qed

\subsection{Proof of Corollary \ref{Cor:MaxPi}}
Differentiating with respect to \(s\), we have  
\[
(\pi_t^*)'(s) = C_{t-1}\gamma\frac{p_t (u_t-1)e^{-\gamma(u_t-1)s} + (1-p_t) (d_t-1)e^{-\gamma(d_t-1)s}}{2\rho_t}.
\]  
Setting \((\pi_t^*)'(s) = 0\), we obtain the equation  
\[
p_t(u_t-1)e^{-\gamma(u_t-1)s} + (1-p_t)(d_t-1)e^{-\gamma(d_t-1)s} = 0.
\]  
We can solve the equation and 
since the second derivative satisfies \((\pi_t^*)''(s) < 0\), it follows that the maximum is attained at \(\pi_t^0\), as defined by \eqref{BenchmarkStrategy}.  

By plugging the maximizer \(\pi_t^0\) back in $\pi_t^*(s)$, we compute the maximum
\begin{align*}
    \pi_t^{*}(\pi_t^0)
    =&\frac{-C_{t-1}p_t\left(\frac{p_t(1-q_t)}{(1-p_t)q_t}\right)^{-(1-q_t)}-C_{t-1}(1-p_t)
    \left(\frac{p_t(1-q_t)}{(1-p_t)q_t}\right)^{q_t}+\beta_t}{2\rho_t}
    \\=&\frac{-C_{t-1}\left(\frac{1-p_t}{1-q_t}\right)^{1-q_t}\left(\frac{p_t}{q_t}\right)^{q_t}+\beta_t}{2\rho_t}.
\end{align*}

Furthermore, since $\frac{\partial}{\partial \rho_t}(\pi_t^*)''(s)
= \frac{C_{t-1}\gamma^{2}}{2\,\rho_t^{2}}
\Big[p_t (u_t-1)^2 e^{-\gamma(u_t-1)s}+(1-p_t)(d_t-1)^2 e^{-\gamma(d_t-1)s}\Big]$ and \((\pi_t^*)''(s) < 0\), the magnitude of the curvature $\vert(\pi_t^*)''\vert$ and hence the concavity increases as $\rho_t$ decreases.

On the other hand, by substituting the expression of $\beta_t$, we have
\begin{align*}
    \pi_t^{*}(\pi_t^0)=\frac{C_{t-1}}{2\rho_t}\left(\frac{\int_{\underline{\pi}_t}^{\overline{\pi}_t}\left(p_te^{-\gamma(u_t-1)s}+(1-p_t)e^{-\gamma(d_t-1)s}\right)ds}{\overline{\pi}_t-\underline{\pi}_t}-\left(\frac{1-p_t}{1-q_t}\right)^{1-q_t}\left(\frac{p_t}{q_t}\right)^{q_t}\right)+\frac{1}{\overline{\pi}_t-\underline{\pi}_t}.
\end{align*}
Define $g^{A}(s)=p_te^{-\gamma(u_t-1)s}+(1-p_t)e^{-\gamma(d_t-1)s}$.
By the mean value theorem, the first term inside the bracket can be expressed as $g^{A}(\tilde{S})$, for some $\tilde{S}\in[\underline{\pi},\overline{\pi}]$, while the second term is the minimum of function $g^{A}(s)$. Therefore, $\pi_t^{*}(\pi_t^0)$ is increasing as $\rho_t$ decreases.
\qed

\subsection{Proof of Lemma \ref{lemma:MonoPiPrime}}
We first compute the partial derivative of $(\pi_t^*)'(s)$ in $\gamma$ and obtain
$$\frac{\partial (\pi_t^{*})^{'}(s)}{\partial \gamma}=C_{t-1}\frac{p_t(u_t-1)e^{-\gamma(u_t-1)s}(1-\gamma (u_t-1)s)+(1-p_t)(d_t-1)e^{-\gamma(d_t-1)s}(1-\gamma (d_t-1)s)}{2\rho_t}.$$
Then for each fixed $s\in [\underline{\pi}_t,\overline{\pi}_t]$, $\frac{\partial (\pi_t^{*})^{'}(s)}{\partial \gamma}<0$ iff
\begin{align}\label{InequalityDerivative}
    p_t(u_t-1)e^{-\gamma(u_t-1)s}(1-\gamma (u_t-1)s)<(1-p_t)(1-d_t)e^{-\gamma(d_t-1)s}(1-\gamma (d_t-1)s).
\end{align}

For any fixed $\gamma\in[\underline{\gamma}_t,\overline{\gamma}_t]$, if $s\ge \frac{1}{\gamma\,(u_t-1)}$ then $1-\gamma(u_t-1)s\le 0$, and the LHS $\le 0$, while the RHS $>0$. Hence (\ref{InequalityDerivative}) holds automatically.
For $0< s< \tfrac{1}{\gamma (u_t-1)}$,
both sides are positive, so take logs of the ratio
\begin{align*}
    \Phi(s,\gamma)&=\log\frac{p\,(u_t-1)\,e^{-\gamma (u_t-1) s}(1-\gamma (u_t-1) s)}{(1-p)\,(1-d_t)\,e^{\gamma (1-d_t) s}(1+\gamma (1-d_t) s)}\\&=\log\frac{p (u_t-1)}{(1-p)(1-d_t)}-\gamma (u_t-d_t)s+\log(1-\gamma (u_t-1) s)-\log(1+\gamma (1-d_t) s).
\end{align*}
Then
$
\frac{\partial \Phi(s,\gamma)}{\partial s}=-\gamma\!\left(u_t-d_t+\frac{u_t-1}{1-\gamma (u_t-1) s}+\frac{1-d_t}{1+\gamma (1-d_t) s}\right)<0.
$ 
Thus $\Phi$ is strictly decreasing on $s\in[0,\tfrac{1}{\gamma (u_t-1)})$ and (\ref{InequalityDerivative}) holds if $\Phi(\underline{\pi}_t,\gamma)\le 0$.

Analogously, for any fixed $s\in[\underline{\pi}_t,\overline{\pi}_t]$, (\ref{InequalityDerivative}) holds automatically if $\gamma \geq \frac{1}{s(u_t-1)}$,
while if $0< \gamma <\frac{1}{s(u_t-1)}$, we have
$\frac{\partial \Phi(s,\gamma)}{\partial \gamma}<0$ for $\gamma\in[0,\tfrac{1}{s (u_t-1)})$ and (\ref{InequalityDerivative}) holds if $\Phi(s,\underline{\gamma}_t)\le 0$.

By plugging in the expressions of $\underline{\pi}_t$, we have
$\Phi(\underline{\pi}_t,\underline{\gamma}_t)=\log\frac{1- (1-q_t)\frac{\underline{\gamma}_t}{\overline{\gamma}_t}\log \frac{p_t(1-q_t)}{(1-p_t)q_t}}{1+ q_t\frac{\underline{\gamma}_t}{\overline{\gamma}_t}\log \frac{p_t(1-q_t)}{(1-p_t)q_t}}+(1-\frac{\underline{\gamma}_t}{\overline{\gamma}_t})\log \frac{p_t(1-q_t)}{(1-p_t)q_t}$, which is strictly negative under Assumption \eqref{Assum:Lemma1}.

Combining the above, we conclude that (\ref{InequalityDerivative}) holds for any $\gamma\in[\underline{\gamma}_t,\overline{\gamma}_t]$ and $s\in[\underline{\pi}_t,\overline{\pi}_t]$ when
$
\frac{\underline{\gamma}_t}{\overline{\gamma}_t}(1-q_t)\,\log\!\left(\frac{p_t(1-q_t)}{(1-p_t)q_t}\right)>1;
$
if this condition is violated, it still holds provided $\Phi(\underline{\pi}_t,\underline{\gamma}_t)< 0$.
This then proves the claim.

Finally, since $\pi_t^*(s)$ integrates to 1 over the interval $s\in[\underline{\pi}_t, \overline{\pi}_t]$ regardless of the value of $\gamma$, then by the monotonicity of $(\pi_t^*)'(s)$ in $\gamma$ and the monotonicity of $\pi_t^*(s)$ in $s$, it readily follows that $\pi_t^*(\underline{\pi}_t)$ is increasing in $\gamma$, while $\pi_t^*(\overline{\pi}_t)$ is decreasing in $\gamma$, and there must exist at least an intersection point at which $\pi_t^{*,a}(s)=\pi_t^{*,(n)}(s)$ on the interval $[\underline{\pi}_t, \overline{\pi}_t]$.
To show the uniqueness of the solution, we denote the first intersection point of $\pi_t^{*,a}(s)$ and $\pi_t^{*,(n)}(s)$ by $s_1$. Assume that $\gamma^{a} < \gamma^{(n)}$, the case when $\gamma^{a} > \gamma^{(n)}$ can be extended analogously, then we have $S_{\gamma^{a}} > S_{\gamma^{(n)}}$ and $(\pi_t^{*,a})^{'}(s)>(\pi_t^{*,(n)})^{'}(s)$. It follows that $\pi_t^{*,a}(s)>\pi_t^{*,(n)}(s)$ for all $s\in(s_1,\overline{\pi}_t]$.
\qed

\subsection{Proof of Proposition \ref{Prop: AcceptanceProb.c}}

We prove the case of $\gamma^a<\gamma^{(n)}$ first. 
By Corollary \ref{Cor:MaxPi}, the maximizer of functions $\pi_t^{*,a}(s)$ and $\pi_t^{*,(n)}(s)$ are resepctively
$S_{\gamma^{a}}$ and $S_{\gamma^{(n)}}.$ When $\gamma^a<\gamma^{(n)}$, we have $S_{\gamma^{a}} > S_{\gamma^{(n)}}$.

We have from Lemma \ref{lemma:MonoPiPrime} that $(\pi_t^{*,a})'(s)>(\pi_t^{*,(n)})'(s)$, $\pi_t^{*,a}(\underline{\pi}_t)<\pi_t^{*,(n)}(\underline{\pi}_t)$ and $\pi_t^{*,a}(\overline{\pi}_t)>\pi_t^{*,(n)}(\overline{\pi}_t)$.
Then the supremum must be attained after the intersection point between the two density functions, which is denoted by $s_1$ and satisfies $\pi_t^{*,a}(s_1)=\pi_t^{*,(n)}(s_1)$.

Note that $0<\pi_t^{*,(n)}(s)<\pi_t^{*,a}(s)$ on $s\in[s_1,\overline{\pi}_t]$, and there are 3 cases that need to be discussed. Case 1: If $s_1\in[S_{\gamma^{a}},\overline{\pi}_t]$, then
$(\pi_t^{*,(n)})'(s)<(\pi_t^{*,a})'(s)<0$ on $s\in[s_1,\overline{\pi}_t]$ and hence $\frac{\partial}{\partial s}\left(\frac{\pi_t^{*,a}(s)}{\pi_t^{*,(n)}(s)}\right)=\frac{(\pi_t^{*,a})'(s)\pi_t^{*,(n)}(s)-\pi_t^{*,a}(s)(\pi_t^{*,(n)})'(s)}{(\pi_t^{*,(n)}(s))^2}>0$.
Case 2: If $s_1\in[S_{\gamma^{(n)}}, S_{\gamma^{a}}]$, then $\frac{\pi_t^{*,a}(s)}{\pi_t^{*,(n)}(s)}$ is strictly increasing on $s\in[S_{\gamma^{a}},\overline{\pi}_t]$ by the same argument in Case 1; while for $s\in[s_1,S_{\gamma^{a}}]$, since $(\pi_t^{*,(n)})'(s)<0<(\pi_t^{*,a})'(s)$, it can be easily shown that $\frac{\partial}{\partial s}\left(\frac{\pi_t^{*,a}(s)}{\pi_t^{*,(n)}(s)}\right)>0$. Case 3: If $s_1\in[\underline{\pi}_t, S_{\gamma^{(n)}}]$, by the same argument as in Case 2, it follows that $\frac{\partial}{\partial s}\left(\frac{\pi_t^{*,a}(s)}{\pi_t^{*,(n)}(s)}\right)>0$ on $s\in[S_{\gamma^{(n)}}, \overline{\pi}_t]$; while for $s\in[s_1,S_{\gamma^{(n)}})$, since $\frac{\pi_t^{*,a}(s)}{\pi_t^{*,(n)}(s)}=\frac{\pi_t^{*,a}(s_1)+\int_{s_1}^s(\pi_t^{*,a})'(r)dr}{\pi_t^{*,(n)}(s_1)+\int_{s_1}^s(\pi_t^{*,(n)})'(r)dr }$, write $F(s)=\int_{s_1}^s(\pi_t^{*,a})'(r)dr$ and $G(s)=\int_{s_1}^s(\pi_t^{*,(n)})'(r)dr$
we have 
\begin{align*}
    \frac{\partial}{\partial s}\left(\frac{\pi_t^{*,a}(s)}{\pi_t^{*,(n)}(s)}\right)&=\frac{(\pi_t^{*,a})'(s)(\pi_t^{*,(n)}(s_1)+ G(s))-(\pi_t^{*,(n)})'(s)(\pi_t^{*,a}(s_1)+F(s))}{(\pi_t^{*,(n)}(s_1)+G(s) )^2}\\&=\frac{((\pi_t^{*,a})'(s)-(\pi_t^{*,(n)})'(s))\pi_t^{*,(n)}(s_1)
    +(\pi_t^{*,a})'(s)G(s) -(\pi_t^{*,(n)})'(s)F(s)}{(\pi_t^{*,(n)}(s_1)+G(s) )^2}
\end{align*}
It follows from $0<(\pi_t^{*,(n)})'(s)<(\pi_t^{*,a})'(s)$ that $((\pi_t^{*,a})'(s)-(\pi_t^{*,(n)})'(s))\pi_t^{*,(n)}(s_1)>0$. On the other hand, $(\pi_t^{*,a})'(s)G(s) -(\pi_t^{*,(n)})'(s)F(s)=G(s)^2\left(\frac{F}{G}\right)^{'}(s)$, then $\frac{\partial}{\partial s}\left(\frac{\pi_t^{*,a}(s)}{\pi_t^{*,(n)}(s)}\right)>0$ iff $$\left(\frac{F}{G}\right)^{'}(s)>-\frac{((\pi_t^{*,a})'(s)-(\pi_t^{*,(n)})'(s))\pi_t^{*,(n)}(s_1)}{G(s)^2}.$$ Note that the RHS is strictly negative, we show in the following that $\left(\frac{F}{G}\right)^{'}(s)>0$, and conclude with $\frac{\partial}{\partial s}\left(\frac{\pi_t^{*,a}(s)}{\pi_t^{*,(n)}(s)}\right)>0$ in all cases, hence the supremum is attained at $\overline{\pi}_t$;

Indeed, by the monotone L'Hôpital rule, if $\frac{(\pi_t^{*,a})'(s)}{(\pi_t^{*,(n)})'(s)}$ is increasing, then $\frac{F(s)}{G(s)}$ is also increasing. Let $\varphi(s)=\frac{(\pi_t^{*,a})'(s)}{(\pi_t^{*,(n)})'(s)}$, the derivative of $\log\varphi(s)$ has the same sign as $\varphi(s)$ because $\varphi>0$.
\begin{align*}
    (\log\varphi)^{'}(s)=-\gamma^{a}\frac{B_{\gamma^{a}}(s)}{A_{\gamma^{a}}(s)}+\gamma^{(n)}\frac{B_{\gamma^{(n)}}(s)}{A_{\gamma^{(n)}}(s)},
\end{align*}
where $A_{\gamma}(s)=p_t(u_t-1)e^{-\gamma (u_t-1)s}+(1-p_t)(d_t-1)e^{-\gamma (d_t-1)s}$, and $B_{\gamma}(s)=p_t(u_t-1)^2e^{-\gamma (u_t-1)s}+(1-p_t)(d_t-1)^2e^{-\gamma (d_t-1)s}$.
Differentiate $\frac{B_{\gamma}}{A_{\gamma}}$ over $\gamma$: $
\frac{\partial}{\partial \gamma}\Big(\frac{B_{\gamma}}{A_{\gamma}}\Big)=\frac{s\,(B_{\gamma}^{2}-A_{\gamma}C_{\gamma})}{A_{\gamma}^{2}},
$ where $C_{\gamma}=p_t(u_t-1)^{3}e^{-\gamma (u_t-1) s}+(1-p_t)(d_t-1)^{3}e^{-\gamma (d_t-1) s}$. The sign is that of $B^{2}-AC$, which can be expanded and simplified to
\begin{align*}
B^{2}-AC
=-\,p_t(1-p_t)\,e^{-\gamma(u_t+d_t-2)s}\,(u_t-1) (d_t-1)\,(u_t-d_t)^{2}>0.
\end{align*}
Therefore, $\frac{\partial}{\partial \gamma}\Big(\frac{B_{\gamma}}{A_{\gamma}}\Big)>0$. Moreover, since $(\pi_t^{*,a})^{'}(s)>0$ and $(\pi_t^{*,(n)})^{'}(s)>0$ for $s\in[s_1,S_{\gamma^{(n)}})$, we have
$\frac{B_{\gamma^{a}}(s)}{A_{\gamma^{a}}}>0$ and $\frac{B_{\gamma^{(n)}}(s)}{A_{\gamma^{(n)}}(s)}>0$, which yields
$(\log\varphi)^{'}>0$.

The case of $\gamma^a>\gamma^{(n)}$ follows directly by applying analogous steps.
\qed

\subsection{Proof of Proposition \ref{Prop:MonotonicityAcceptProb}}
We only prove for the case of $\gamma^{a}<\gamma^{(n)}$ herein. The case when $\gamma^{a}>\gamma^{(n)}$ can be extended in a similar manner.
We first compute the derivative of $g_1(\gamma^{(n)},\gamma^{a})$ w.r.t. $\gamma^{(n)}$.
\begin{equation*}
    \begin{aligned}
    \frac{\partial g_1(\gamma^{(n)},\gamma^{a})}{\partial \gamma^{(n)}}=\log\left(\frac{p_t(1-q_t)}{(1-p_t)q_t}\right)\frac{\gamma^{a}C_{t-1}}{(\gamma^{(n)})^2}\Bigg(q_t&(1-p_t)\left(\frac{p_t(1-q_t)}{(1-p_t)q_t}\right)^{\frac{\gamma^{a}}{\gamma^{(n)}}q_t}\\&-(1-q_t)p_t\left(\frac{p_t(1-q_t)}{(1-p_t)q_t}\right)^{-\frac{\gamma^{a}}{\gamma^{(n)}}(1-q_t)}
    \Bigg).
\end{aligned}
\end{equation*}
Since the market offers positive expected excess return, it can be readily shown that the term in the bracket, and hence the derivative $\frac{\partial g_1(\gamma^{(n)},\gamma^{a})}{\partial \gamma^{(n)}}$, is strictly negative when $\gamma^{(n)}>\gamma^{a}$.

We next discuss the derivative of $h(\gamma^{(n)})$ in $\gamma^{(n)}$. The numerator of $\frac{\partial h(\gamma^{(n)})}{\partial \gamma^{(n)}}$ is computed by
\begin{align*}
    N(\gamma^{(n)})=&
    \overline{\pi}_tC_{t-1}\left((u_t-1)p_te^{-\gamma^{(n)}(u_t-1)\overline{\pi}_t}+(d_t-1)(1-p_t)e^{-\gamma^{(n)}(d_t-1)\overline{\pi}_t}\right)\\ & \hspace{1cm} \times
    \left(-C_{t-1}\left(\frac{1-p_t}{1-q_t}\right)^{1-q_t}\left(\frac{p_t}{q_t}\right)^{q_t}+\beta_t^{(n)}\right)
    \\&+\frac{\partial \beta_t^{(n)}}{\partial \gamma^{(n)}}\left(C_{t-1}p_te^{-\gamma^{(n)}(u_t-1)\overline{\pi}_t}+C_{t-1}(1-p_t)e^{-\gamma^{(n)}(d_t-1)\overline{\pi}_t}-C_{t-1}\left(\frac{1-p_t}{1-q_t}\right)^{1-q_t}\left(\frac{p_t}{q_t}\right)^{q_t}\right).
\end{align*}
We know that $C_{t-1}p_te^{-\gamma^{(n)}(u_t-1)\overline{\pi}_t}+C_{t-1}(1-p_t)e^{-\gamma^{(n)}(d_t-1)\overline{\pi}_t}-C_{t-1}\left(\frac{1-p_t}{1-q_t}\right)^{1-q_t}\left(\frac{p_t}{q_t}\right)^{q_t}>0$ and $(u_t-1)p_te^{-\gamma^{(n)}(u_t-1)\overline{\pi}_t}+(d_t-1)(1-p_t)e^{-\gamma^{(n)}(d_t-1)\overline{\pi}_t}<0$ from Corollary \ref{Cor:MaxPi}. Furthermore,
by the positivity of the probability density, we have $-C_{t-1}\left(\frac{1-p_t}{1-q_t}\right)^{1-q_t}\left(\frac{p_t}{q_t}\right)^{q_t}+\beta_t^{(n)}>0$, and $\beta_t^{(n)}>C_{t-1}p_te^{-\gamma^{(n)}(u_t-1)\overline{\pi}_t}+C_{t-1}(1-p_t)e^{-\gamma^{(n)}(d_t-1)\overline{\pi}_t}$.


We next show 
\begin{align}\label{InequalityMagnitude}
    \overline{\pi}_tC_{t-1}\left((u_t-1)p_te^{-\gamma^{(n)}(u_t-1)\overline{\pi}_t}+(d_t-1)(1-p_t)e^{-\gamma^{(n)}(d_t-1)\overline{\pi}_t}\right)<-\frac{\partial \beta_t^{(n)}}{\partial \gamma^{(n)}}.
\end{align}

By the mean-value theorem,
\begin{align*}
    \frac{\partial \beta_t^{(n)}}{\partial \gamma^{(n)}}=&\frac{-C_{t-1}}{\overline{\pi}_t-\underline{\pi}_t}\int_{\underline{\pi}_t}^{\overline{\pi}_t}s\left((u_t-1)p_te^{-\gamma^{(n)}(u_t-1)s}+(d_t-1)(1-p_t)e^{-\gamma^{(n)}(d_t-1)s}\right)ds\\=&-\tilde{S}C_{t-1}\left((u_t-1)p_te^{-\gamma^{(n)}(u_t-1)\tilde{S}}+(d_t-1)(1-p_t)e^{-\gamma^{(n)}(d_t-1)\tilde{S}}\right)
\end{align*}
for some $\tilde{S}\in[\underline{\pi},\overline{\pi}]$.
Consider function $f(s)=(u_t-1)p_te^{-\gamma^{(n)}(u_t-1)s}+(d_t-1)(1-p_t)e^{-\gamma^{(n)}(d_t-1)s}$, then \eqref{InequalityMagnitude} becomes
$\frac{\overline{\pi}_t}{\tilde{S}}f(\overline{\pi}_t)<f(\tilde{S}).$
Since $f(s)$ is strictly decreasing in $s$, we have $f(\overline{\pi}_t)\leq f(\tilde{S}).$

Moreover,
$f(\overline{\pi}_t)<0$ and $\frac{\overline{\pi}_t}{\tilde{S}}>1$, then
    $\frac{\overline{\pi}_t}{\tilde{S}}f(\overline{\pi}_t)<f(\overline{\pi}_t)\leq f(\tilde{S}).$
This establishes \eqref{InequalityMagnitude}, from which it further follows that $N(\gamma^{(n)})<0$. Since $g_1(\gamma^{(n)},\gamma^{a})$, $g_2(\gamma^{a})$, and $h(\gamma^{(n)})$ are all positive, the derivative of $\frac{\pi^{*,a}(S_{\gamma^{(n)}})}{c_t^{(n)}\,\pi^{*,(1)}(S_{\gamma^{(n)}})}$ with respect to $\gamma^{(n)}$ is strictly negative on $[\gamma^{a},\overline{\gamma}_t]$.

To show the second part of the proposition is equivalent to proving that 
$\frac{g_1(\gamma^{(n)},\gamma^{a})}{g_2(\gamma^{a})}$ is strictly increasing w.r.t. $\gamma^{a}$ when assuming $\gamma^{a}<\gamma^{(n)}$.
Given that 
$$\frac{g_1(\gamma^{(n)},\gamma^{a})}{g_2(\gamma^{a})}=\frac{-C_{t-1}p_te^{-\gamma^a(u_t-1)S_{\gamma^{(n)}}}-C_{t-1}(1-p_t)e^{-\gamma^a(d_t-1)S_{\gamma^{(n)}}}+\beta_t^{a}}{-C_{t-1}p_te^{-\gamma^a(u_t-1)\overline{\pi}_t}-C_{t-1}(1-p_t)e^{-\gamma^a(d_t-1)\overline{\pi}_t}+\beta_t^{a}},$$ and that $S_{\gamma^{a}}>S_{\gamma^{(n)}}$ when $\gamma^{a}<\gamma^{(n)}$, $\frac{\partial (\pi_t^{*})^{'}(s)}{\partial \gamma}<0$ from Lemma \ref{lemma:MonoPiPrime} can be further applied to complete the proof.
\qed

\subsection{Proof of Theorem \ref{Thm:ExploForwardRL-N=2}}
We first present an auxiliary result on the construction of PreFER processes under CARA preferences with two risky assets.
\begin{lemma}\label{LemmaSup-N=2}
Suppose that \( (f_t^{(1)}, f_t^{(2)}) \) is the random control process associated with the probability density \(\pi_t\). For any $t\in\mathbb{N}$, we have
\begin{equation}\label{SupremumEqn-N=2}
\begin{aligned}
    &\sup_{\pi_t\in \mathcal{A}(t)}\E_{\pi}\left[-C_{t-1}
    e^{-\gamma \sum\limits_{i=1}^2 f_t^{(i)}(R_t^{(i)}-1)}-\rho_t\int_{\underline{\pi}_t^{(1)}}^{\overline{\pi}_t^{(1)}}\int_{\underline{\pi}_t^{(2)}}^{\overline{\pi}_t^{(2)}}\pi_t^2(s^{(1)},s^{(2)})ds^{(1)}ds^{(2)}\Bigg\vert \mathcal{F}_{t-1}\right]\\=&\frac{1}{4\rho_t}\int_{\underline{\pi}_t^{(1)}}^{\overline{\pi}_t^{(1)}}\int_{\underline{\pi}_t^{(2)}}^{\overline{\pi}_t^{(2)}}h_t^2(s^{(1)},s^{(2)})ds^{(1)}ds^{(2)}-\beta_t.
\end{aligned}
\end{equation}
\end{lemma}
\begin{proof}
    Set the Lagrangian function
\begin{align*}
L(\pi_t,\beta_t)=&\int_{\underline{\pi}_t^{(1)}}^{\overline{\pi}_t^{(1)}} \int_{\underline{\pi}_t^{(2)}}^{\overline{\pi}_t^{(2)}}\left(E\left[-C_{t-1}e^{-\gamma \sum\limits_{i=1}^2 s^{(i)}(R_t^{(i)}-1)}\Big\vert\mathcal{F}_{t-1}\right]-\rho_t \pi_t(s^{(1)},s^{(2)})\right)\pi_t(s^{(1)},s^{(2)})ds^{(1)}ds^{(2)}\\\qquad &+\beta_t \left(\int_{\underline{\pi}^{(1)}}^{\overline{\pi}_t^{(1)}}\int_{\underline{\pi}_t^{(2)}}^{\overline{\pi}_t^{(2)}}\pi_t(s^{(1)},s^{(2)})ds^{(1)}ds^{(2)}-1\right).
\end{align*}
Then by the method of Lagrange multipliers, we solve the associated system of equations:
$$\int_{\underline{\pi}_t^{(1)}}^{\overline{\pi}_t^{(1)}} \int_{\underline{\pi}_t^{(2)}}^{\overline{\pi}_t^{(2)}}\left(E\left[-C_{t-1}e^{-\gamma\sum\limits_{i=1}^N s^{(i)}(R^{(i)}-1)}\Big\vert\mathcal{F}_{t-1}\right]-2\rho_t \pi_t (s^{(1)},s^{(2)})+\beta_t\right)ds^{(1)}ds^{(2)}=0,$$
$$\int_{\underline{\pi}_t^{(1)}}^{\overline{\pi}_t^{(1)}}\int_{\underline{\pi}_t^{(2)}}^{\overline{\pi}_t^{(2)}}\pi_t(s^{(1)},s^{(2)})ds^{(1)}ds^{(2)}-1=0,$$ and obtain the optimal control whose density function is given by
$$\pi_t^{*}(s^{(1)},s^{(2)})=\frac{1}{2\rho_t}\left(E\left[-C_{t-1}e^{-\gamma\sum\limits_{i=1}^N s^{(i)}(R^{(i)}-1)}\Bigg\vert\mathcal{F}_{t-1}\right]+\beta_t\right),$$ with $\beta_t=\frac{2\rho_t-\int_{\underline{\pi}_t^{(1)}}^{\overline{\pi}_t^{(1)}} \int_{\underline{\pi}_t^{(2)}}^{\overline{\pi}_t^{(2)}}E\left[-C_{t-1}e^{-\gamma\sum\limits_{i=1}^2 s^{(i)}(R^{(i)}-1)}\Big\vert\mathcal{F}_{t-1}\right]ds^{(1)}ds^{(2)}}{
(\overline{\pi}_t^{(1)}-\underline{\pi}_t^{(1)})(\overline{\pi}_t^{(2)}-\underline{\pi}_t^{(2)})}$.

Additionally, to satisfy the requirements of a probability density function, it is crucial that \(\pi_t(s^{(1)},s^{(2)}) \geq 0\). This condition is guaranteed by Assumption \eqref{AssumptionForPositive-N=2}.

Next, by plugging $\pi_t^{*}(s^{(1)},s^{(2)})$ back in the LHS of \eqref{SupremumEqn-N=2}, it becomes
\begin{align*}
&\frac{1}{4\rho_t}\int_{\underline{\pi}_t^{(1)}}^{\overline{\pi}_t^{(1)}}\int_{\underline{\pi}_t^{(2)}}^{\overline{\pi}_t^{(2)}}\left(E\left[-C_{t-1}e^{-\gamma\sum\limits_{i=1}^2 s^{(i)}(R^{(i)}-1)}\Big\vert\mathcal{F}_{t-1}\right]+\beta_t\right)^2ds^{(1)}ds^{(2)}\\ & -\frac{\beta_t}{2\rho_t}\int_{\underline{\pi}_t^{(1)}}^{\overline{\pi}_t^{(1)}}\int_{\underline{\pi}_t^{(2)}}^{\overline{\pi}_t^{(2)}}\left(E\left[-C_{t-1}e^{-\gamma\sum\limits_{i=1}^2 s^{(i)}(R^{(i)}-1)}\Big\vert\mathcal{F}_{t-1}\right]+\beta_t\right)ds^{(1)}ds^{(2)}\\=&\frac{1}{4\rho_t}\int_{\underline{\pi}_t^{(1)}}^{\overline{\pi}_t^{(1)}}\int_{\underline{\pi}_t^{(2)}}^{\overline{\pi}_t^{(2)}}\left(E\left[-C_{t-1}e^{-\gamma\sum\limits_{i=1}^2 s^{(i)}(R^{(i)}-1)}\Big\vert\mathcal{F}_{t-1}\right]+\beta_t\right)^2ds^{(1)}ds^{(2)}-\beta_t.
\end{align*}

On the other hand, Since for any $s^{(1)}\in [\underline{\pi}_t^{(1)}, \overline{\pi}_t^{(1)}]$ and $s^{(2)}\in [\underline{\pi}_t^{(2)}, \overline{\pi}_t^{(2)}]$,
\begin{align*}
    &\E\left[e^{-\gamma \sum\limits_{i=1}^2 s^{(i)}(R_t^{(i)}-1)}\Big\vert \mathcal{F}_{t-1}\right]
    \\=&\E\left[\E\left[e^{-\gamma s^{(1)}(R_t^{(1)}-1)}\Big\vert \mathcal{F}_{t-1}\vee \sigma(\xi_t^{cn})\right]\E\left[e^{-\gamma s^{(2)}(R_t^{(2)}-1)}\Big\vert \mathcal{F}_{t-1}\vee \sigma(\xi_t^{cn})\right]\Bigg\vert \mathcal{F}_{t-1}\right]\\=&p_t^{cn}\prod\limits_{i=1}^2\left(p_t^{(i),1}e^{-\gamma s^{(i)}(u_t^{(i)}-1)}+(1-p_t^{(i),1})e^{-\gamma s^{(i)}(d_t^{(i)}-1)}\right)\\&+(1-p_t^{cn})\prod\limits_{i=1}^2\left(p_t^{(i),0}e^{-\gamma s^{(i)}(u_t^{(i)}-1)}+(1-p_t^{(i),0})e^{-\gamma s^{(i)}(d_t^{(i)}-1)}\right),
\end{align*}
where the last line holds because for $k=0,1$ and $i=1,2$, $$\E\left[e^{-\gamma  s^{(i)}(R_t^{(i)}-1)}\Big\vert \mathcal{F}_{t-1}\vee \sigma(\xi_t^{cn})\right]\mathbbm{1}_{\{\xi_t^{cn}=k\}}=p_t^{(i),k}e^{-\gamma s^{(i)}(u_t^{(i)}-1)}+(1-p_t^{(i),k})e^{-\gamma s^{(i)}(d_t^{(i)}-1)}.$$
This then proves the claim \eqref{SupremumEqn-N=2} by letting $h_t(s^{(1)},s^{(2)})$ be given by (\ref{Expression-h-N=2}).
\end{proof}

We continue with the proof of Theorem \ref{Thm:ExploForwardRL-N=2}.
    Conditions $(i)$ and $(ii)$ of Definition \ref{def:ExploratoryFrowardPreferences} follow directly.
Indeed, let $x \in \mathbb{R}$, we first note that for any $t\in \mathbb{N}$, $s^{(1)}\in[\underline{\pi}^{(2)}_t,\overline{\pi}^{(2)}_t]$ and $s^{(2)}\in[\underline{\pi}^{(2)}_t,\overline{\pi}^{(2)}_t]$, the sequence of $(J_n)_{n=1}^t$ are Borel-measurable functions of the market parameters $(p_n^{(i),k})_{n=1}^t, (u_n^{(i)})_{n=1}^t,(d_n^{(i)})_{n=1}^t$ for $i=1,2$ and $k=0,1$, and $(p_n^{cn})_{n=1}^t, (\rho_n)_{n=1}^t$ which are $\mathcal{F}_{t-1}$-measurable. 
It then follows that the process $U_t(x)=-C_te^{-\gamma x}$ is $\mathcal{F}_{t-1}$-measurable.

Next, we show Condition $(iii)$ of Definition \ref{def:ExploratoryFrowardPreferences}. 
Let $t\in \mathbb{N}$ and consider any arbitrary ${X}\in {\mathcal{X}}(x)$, we must show that
\begin{align*}
    \E_{\pi}\left[U_{t}(X_{t})\vert \mathcal{F}_{t}\right]\geq \E_{\pi}\left[U_{t+1}(X_{t+1})-\rho_{t+1} P_{t+1}e^{-\gamma X_{t}}\int_{\underline{\pi}^{(1)}_t}^{\overline{\pi}^{(1)}_t}\int_{\underline{\pi}^{(2)}_t}^{\overline{\pi}^{(2)}_t}\pi_{t+1}^2(s^{(1)},s^{(2)})ds^{(1)}ds^{(2)}\bigg\vert \mathcal{F}_{t}\right].
\end{align*}

Straightforward computation shows
\begin{align*}
    &\E_{\pi}\left[U_{t+1}(X_{t+1})-\rho_{t+1} P_{t+1}e^{-\gamma X_{t}}\int_{\underline{\pi}^{(1)}_t}^{\overline{\pi}^{(1)}_t}\int_{\underline{\pi}^{(2)}_t}^{\overline{\pi}^{(2)}_t}\pi_{t+1}^2(s^{(1)},s^{(2)})ds^{(1)}ds^{(2)}\bigg\vert \mathcal{F}_{t}\right]\\=&P_{t+1}\E_{\pi}\left[e^{-\gamma X_{t}} \vert \mathcal{F}_{t}\right]
    \E_{\pi}\left[-C_te^{-\gamma \sum\limits_{i=1}^2 f_t^{(i)}(R_t^{(i)}-1)}-\rho_{t+1} \int_{\underline{\pi}^{(1)}_t}^{\overline{\pi}^{(1)}_t}\int_{\underline{\pi}^{(2)}_t}^{\overline{\pi}^{(2)}_t}\pi_{t+1}^2(s^{(1)},s^{(2)})ds^{(1)}ds^{(2)}\bigg\vert \mathcal{F}_{t}\right].
\end{align*}

Furthermore, $U_t(X_t)=-C_te^{-\gamma X_{t}}$, and $P_{t+1}=C_tJ_{t+1}$. Hence, it remains to show that
\begin{align*}
    -1\geq J_{t+1}\E_{\pi}\left[-C_te^{-\gamma \sum\limits_{i=1}^2 f_t^{(i)}(R_t^{(i)}-1)}-\rho_{t+1} \int_{\underline{\pi}^{(1)}_t}^{\overline{\pi}^{(1)}_t}\int_{\underline{\pi}^{(2)}_t}^{\overline{\pi}^{(2)}_t}\pi_{t+1}^2(s^{(1)},s^{(2)})ds^{(1)}ds^{(2)}\bigg\vert \mathcal{F}_{t}\right],
\end{align*}
which easily follows using Lemma \ref{LemmaSup-N=2}.

To show Condition $(iv)$, we let $X_t^{*}$, $t \in \mathbb{N}_0$, be given by (\ref{OptimalWealthRL}). We need to establish
\begin{align*}
    \E_{\pi}\left[U_t(X_t^{*})\vert \mathcal{F}_{t}\right]
    = \E_{\pi}\left[U_{t+1}(X_{t+1}^{*})-\rho_{t+1} P_{t+1}e^{-\gamma X_{t}^*}\int_{\underline{\pi}_t}^{\overline{\pi}_t}(\pi_{t+1}^*)^2(s^{(1)},s^{(2)})ds^{(1)}ds^{(2)}\bigg\vert \mathcal{F}_{t}\right].
\end{align*}
Since ${X}_{t+1}^{*}={X}_{t}^{*}+\sum\limits_{i=1}^2 \pi_t^{(i),*}(R_t^{(i)}-1)$, the above equality reduces to 
\begin{align*}
    -1= J_{t+1}\E_{\pi}\left[-C_te^{-\gamma \sum\limits_{i=1}^2 f_t^{(i),*}(R_t^{(i)}-1)}-\rho_{t+1} \int_{\underline{\pi}^{(1)}_t}^{\overline{\pi}^{(1)}_t}\int_{\underline{\pi}^{(2)}_t}^{\overline{\pi}^{(2)}_t}(\pi_{t+1}^{*})^2(s^{(1)},s^{(2)})ds^{(1)}ds^{(2)}\bigg\vert \mathcal{F}_{t}\right].
\end{align*}
This then directly follows with Lemma \ref{LemmaSup-N=2}.
\qed

\subsection{Proof of Theorem \ref{Thm:ExploForwardRL-CRRA}}
We first present an auxiliary result for constructing PreFER processes within the CRRA family.
\begin{lemma}\label{LemmaSupCRRA}
Suppose that \( f_t \) is the random control process associated with the probability density \(\pi_t\). For any $t\in\mathbb{N}$, we have
\begin{align}\label{SupremumEqn-CRRA}
    \sup_{\pi_t\in \mathcal{A}(t)}\E_{\pi}\left[C_{t-1}\left(1+f_t(R_t-1)\right)^{1-\frac{1}{\theta}}-\rho_t\int_{\underline{\pi}_t}^{\overline{\pi}_t}\pi_t^2(s)ds\Bigg\vert \mathcal{F}_{t-1}\right]=\frac{1}{4\rho_t}\int_{\underline{\pi}_t}^{\overline{\pi}_t}h_t^2(s)ds-\beta_t.
\end{align}
\end{lemma}
\begin{proof}
    We first set the Lagrangian function for the constrained optimization problem,
\begin{align*} L(\pi_t,\beta_t)=&\int_{\underline{\pi}_t}^{\overline{\pi}_t}\left(C_{t-1}\left(p_t\left(1+s(u_t-1)\right)^{1-\frac{1}{\theta}}+(1-p_t)\left(1+s(d_t-1)\right)^{1-\frac{1}{\theta}}\right)-\rho_t \pi_t(s)\right)\pi_t(s)ds\\& \qquad+\beta_t \left(\int_{\underline{\pi}_t}^{\overline{\pi}_t}\pi_t(s)ds-1\right),
\end{align*} 
then by the method of Lagrange multipliers, we obtain the following system of equations:
$$\int_{\underline{\pi}_t}^{\overline{\pi}_t}\left(C_{t-1}p_t\left(1+s(u_t-1)\right)^{1-\frac{1}{\theta}}+C_{t-1}(1-p_t)\left(1+s(d_t-1)\right)^{1-\frac{1}{\theta}}-2\rho_t \pi_t(s)+ \beta_t \right)ds =0,$$
$$\int_{\underline{\pi}_t}^{\overline{\pi}_t}\pi_t(s)ds - 1=0.$$
Applying the usual verification technique, we obtain an optimal control whose density function is given by
$$\pi_t^{*}(s)=\frac{C_{t-1}p_t\left(1+s(u_t-1)\right)^{1-\frac{1}{\theta}}+C_{t-1}(1-p_t)\left(1+s(d_t-1)\right)^{1-\frac{1}{\theta}}+\beta_t}{2\rho_t},$$ with
$\beta_t=\frac{2\rho_t-C_{t-1}\int_{\underline{\pi}_t}^{\overline{\pi}_t}\left(p_t\left(1+s(u_t-1)\right)^{1-\frac{1}{\theta}}+(1-p_t)\left(1+s(d_t-1)\right)^{1-\frac{1}{\theta}}\right)ds}{\overline{\pi}_t-\underline{\pi}_t},$
which is further equal to \eqref{Expression-beta-CRRA} after direct computation.
Additionally, to satisfy the requirements of a probability density function, it is crucial that \(\pi_t(s) \geq 0\). This condition is guaranteed by Assumption \eqref{AssumptionForPositive-CRRA}.

Next, by plugging $\pi_t^{*}(s)$ back in the LHS of \eqref{SupremumEqn-CRRA}, it becomes
\begin{align*}
&\frac{1}{4\rho_t}\int_{\underline{\pi}_t}^{\overline{\pi}_t}(C_{t-1}p_t\left(1+s(u_t-1)\right)^{1-\frac{1}{\theta}}-C_{t-1}(1-p_t)\left(1+s(d_t-1)\right)^{1-\frac{1}{\theta}}+\beta_t)^2ds\\ & -\frac{\beta_t}{2\rho_t}\int_{\underline{\pi}_t}^{\overline{\pi}_t}(C_{t-1}p_t\left(1+s(u_t-1)\right)^{1-\frac{1}{\theta}}-C_{t-1}(1-p_t)\left(1+s(d_t-1)\right)^{1-\frac{1}{\theta}}+\beta_t)ds\\=&\frac{1}{4\rho_t}\int_{\underline{\pi}_t}^{\overline{\pi}_t}(C_{t-1}p_t\left(1+s(u_t-1)\right)^{1-\frac{1}{\theta}}-C_{t-1}(1-p_t)\left(1+s(d_t-1)\right)^{1-\frac{1}{\theta}}+\beta_t)^2ds-\beta_t,
\end{align*}
which further equals to $\frac{1}{4\rho_t}\int_{\underline{\pi}_t}^{\overline{\pi}_t}h_t^2(s)ds-\beta_t$, by letting $h_t(s)$ be given by (\ref{Expression-h-CRRA}).
\end{proof}

We continue with the proof of Theorem \ref{Thm:ExploForwardRL-CRRA}.

    Conditions $(i)$ and $(ii)$ of Definition \ref{def:ExploratoryFrowardPreferences} follow directly.
Indeed, let $x \in \mathbb{R}_+$, we first note that for any $t\in \mathbb{N}$ and $s\in[\underline{\pi}_t,\overline{\pi}_t]$, the sequence of $(J_n)_{n=1}^t$ are Borel-measurable functions of the market parameters $(p_n)_{n=1}^t, (u_n)_{n=1}^t,(d_n)_{n=1}^t, (\rho_n)_{n=1}^t$ which are $\mathcal{F}_{t-1}$-measurable. 
It then follows that the process $U_t(x)=C_t(1-\frac{1}{\theta})x^{1-\frac{1}{\theta}}$ is $\mathcal{F}_{t-1}$-measurable.

Next, we show Condition $(iii)$ of Definition \ref{def:ExploratoryFrowardPreferences}. 
Let $t\in \mathbb{N}$ and consider any arbitrary ${X}\in {\mathcal{X}}(x)$, we must show that
\begin{align*}
    \E_{\pi}\left[U_{t}(X_{t})\vert \mathcal{F}_{t}\right]\geq \E_{\pi}\left[U_{t+1}(X_{t+1})-\rho_{t+1} P_{t+1}(1-\frac{1}{\theta})X_{t}^{1-\frac{1}{\theta}}\int_{\underline{\pi}_t}^{\overline{\pi}_t}\pi_{t+1}^2(s)ds\bigg\vert \mathcal{F}_{t}\right].
\end{align*}

Straightforward computation yields
\begin{align*}
    &\E_{\pi}\left[U_{t+1}(X_{t+1})-\rho_{t+1} P_{t+1}(1-\frac{1}{\theta})X_{t}^{1-\frac{1}{\theta}}\int_{\underline{\pi}_t}^{\overline{\pi}_t}\pi_{t+1}^2(s)ds\bigg\vert \mathcal{F}_{t}\right]\\=&P_{t+1}(1-\frac{1}{\theta})X_{t}^{1-\frac{1}{\theta}}
    \E_{\pi}\left[C_t\left(1+f_{t+1}(R_{t+1}-1)\right)^{1-\frac{1}{\theta}}-\rho_{t+1} \int_{\underline{\pi}_t}^{\overline{\pi}_t}\pi_{t+1}^2(s)ds\bigg\vert \mathcal{F}_{t}\right].
\end{align*}

Furthermore, $U_t(X_t)=C_t(1-\frac{1}{\theta})X_t^{1-\frac{1}{\theta}}$, and $P_{t+1}=C_tJ_{t+1}$. Hence, it remains to show that
\begin{align*}
    1\geq J_{t+1}\E_{\pi}\left[C_t\left(1+f_{t+1}(R_{t+1}-1)\right)^{1-\frac{1}{\theta}}-\rho_{t+1} \int_{\underline{\pi}_t}^{\overline{\pi}_t}\pi_{t+1}^2(s)ds\bigg\vert \mathcal{F}_{t}\right],
\end{align*}
which easily follows using Lemma \ref{LemmaSupCRRA}.

To show Condition $(iv)$, we let $X_t^{*}$, $t \in \mathbb{N}_0$, be given by (\ref{OptimalWealthRL}). We need to establish
\begin{align*}
\E_{\pi}\left[U_t(X_t^{*})\vert \mathcal{F}_{t}\right]
    = \E_{\pi}\left[U_{t+1}(X_{t+1}^{*})-\rho_{t+1} P_{t+1}e^{-\gamma X_{t}^*}\int_{\underline{\pi}_t}^{\overline{\pi}_t}(\pi_{t+1}^*)^2(s)ds\bigg\vert \mathcal{F}_{t}\right].
\end{align*}
Since ${X}_{t+1}^{*}={X}_{t}^{*}(1+f_{t+1}^{*}(R_{t+1}-1))$, the above equality reduces to 
\begin{align*}
    1= J_{t+1}\E_{\pi}\left[C_t\left(1+f_{t+1}(R_{t+1}-1)\right)^{1-\frac{1}{\theta}}-\rho_{t+1} \int_{\underline{\pi}_t}^{\overline{\pi}_t}(\pi_{t+1}^*)^2(s)ds\bigg\vert \mathcal{F}_{t}\right].
\end{align*}
This then directly follows with Lemma \ref{LemmaSupCRRA}.
\qed

\bibliography{BibFile}

@article{gomez2015netflix,
  title={The netflix recommender system: Algorithms, business value, and innovation},
  author={Gomez-Uribe, Carlos A and Hunt, Neil},
  journal={ACM Transactions on Management Information Systems (TMIS)},
  volume={6},
  number={4},
  pages={1--19},
  year={2015},
  publisher={ACM New York, NY, USA}
}

@inproceedings{covington2016deep,
  title={Deep neural networks for youtube recommendations},
  author={Covington, Paul and Adams, Jay and Sargin, Emre},
  booktitle={Proceedings of the 10th ACM Conference on Recommender Systems},
  pages={191--198},
  year={2016}
}

@incollection{ricci2010introduction,
  title={Introduction to recommender systems handbook},
  author={Ricci, Francesco and Rokach, Lior and Shapira, Bracha},
  booktitle={Recommender Systems Handbook},
  pages={1--35},
  year={2010},
  publisher={Springer}
}

@inproceedings{zhao2018deep,
  title={Deep reinforcement learning for page-wise recommendations},
  author={Zhao, Xiangyu and Xia, Long and Zhang, Liang and Ding, Zhuoye and Yin, Dawei and Tang, Jiliang},
  booktitle={Proceedings of the 12th ACM Conference on Recommender Systems},
  pages={95--103},
  year={2018}
}

@inproceedings{chen2019top,
  title={Top-k off-policy correction for a REINFORCE recommender system},
  author={Chen, Minmin and Beutel, Alex and Covington, Paul and Jain, Sagar and Belletti, Francois and Chi, Ed H},
  booktitle={Proceedings of the Twelfth ACM International Conference on Web Search and Data Mining},
  pages={456--464},
  year={2019}
}

@article{cheng2023eliciting,
  title={Eliciting Risk Aversion with Inverse Reinforcement Learning via Interactive Questioning},
  author={Cheng, Ziteng and Coache, Anthony and Jaimungal, Sebastian},
  journal={arXiv preprint arXiv:2308.08427},
  year={2025}
}

@article{angoshtari2020predictable,
  title={Predictable forward performance processes: The binomial case},
  author={Angoshtari, Bahman and Zariphopoulou, Thaleia and Zhou, Xun Yu},
  journal={SIAM Journal on Control and Optimization},
  volume={58},
  number={1},
  pages={327--347},
  year={2020},
  publisher={SIAM}
}

@article{musiela2006investments,
  title={Investments and forward utilities},
  author={Musiela, Marek and Zariphopoulou, Thaleia},
  journal={Preprint},
  year={2006}
}

@article{musiela2008optimal,
  title={Optimal asset allocation under forward exponential performance criteria},
  author={Musiela, Marek and Zariphopoulou, Thaleia},
  journal={Markov Processes and Related Topics: A Festschrift for Thomas G. Kurtz},
  volume={4},
  pages={285--300},
  year={2008},
  publisher={Institute of Mathematical Statistics Beachwood, OH, USA}
}

@article{musiela2009portfolio,
  title={Portfolio choice under dynamic investment performance criteria},
  author={Musiela, Marek and Zariphopoulou, Thaleia},
  journal={Quantitative Finance},
  volume={9},
  number={2},
  pages={161--170},
  year={2009},
  publisher={Taylor \& Francis}
}

@article{musiela2010portfolio,
  title={Portfolio choice under space-time monotone performance criteria},
  author={Musiela, Marek and Zariphopoulou, Thaleia},
  journal={SIAM Journal on Financial Mathematics},
  volume={1},
  number={1},
  pages={326--365},
  year={2010},
  publisher={SIAM}
}

@article{henderson2007horizon,
  title={Horizon-unbiased utility functions},
  author={Henderson, Vicky and Hobson, David},
  journal={Stochastic Processes and their Applications},
  volume={117},
  number={11},
  pages={1621--1641},
  year={2007},
  publisher={Elsevier}
}

@article{choulli2007minimal,
  title={Minimal Hellinger martingale measures of order q},
  author={Choulli, Tahir and Stricker, Christophe and Li, Jia},
  journal={Finance and Stochastics},
  volume={11},
  number={3},
  pages={399--427},
  year={2007},
  publisher={Springer}
}

@article{kallblad2018dynamically,
  title={Dynamically consistent investment under model uncertainty: the robust forward criteria},
  author={K{\"a}llblad, Sigrid and Ob{\l}{\'o}j, Jan and Zariphopoulou, Thaleia},
  journal={Finance and Stochastics},
  volume={22},
  number={4},
  pages={879--918},
  year={2018},
  publisher={Springer}
}

@article{källblad2020black,
  title={Black's Inverse Investment Problem and Forward Criteria with Consumption},
  author={Källblad, Sigrid},
  journal={SIAM Journal on Financial Mathematics},
  volume={11},
  number={2},
  pages={494--525},
  year={2020},
  publisher={SIAM}
}

@article{hu2020systems,
  title={Systems of ergodic {BSDE}s arising in regime switching forward performance processes},
  author={Hu, Ying and Liang, Gechun and Tang, Shanjian},
  journal={SIAM Journal on Control and Optimization},
  volume={58},
  number={4},
  pages={2503--2534},
  year={2020},
  publisher={SIAM}
}

@article{he2021forward,
  title={Forward rank-dependent performance criteria: Time-consistent investment under probability distortion},
  author={He, Xue Dong and Strub, Moris S and Zariphopoulou, Thaleia},
  journal={Mathematical Finance},
  volume={31},
  number={2},
  pages={683--721},
  year={2021},
  publisher={Wiley Online Library}
}

@article{jia2022policya,
  title={Policy evaluation and temporal-difference learning in continuous time and space: A martingale approach},
  author={Jia, Yanwei and Zhou, Xun Yu},
  journal={Journal of Machine Learning Research},
  volume={23},
  number={154},
  pages={1--55},
  year={2022}
}

@article{cui2023discrete,
  title={Discrete-Time Mean-Variance Strategy Based on Reinforcement Learning},
  author={Cui, Xiangyu and Li, Xun and Shi, Yun and Zhao, Si},
  journal={arXiv preprint arXiv:2312.15385},
  year={2023}
}

@article{jia2023q,
  title={q-Learning in continuous time},
  author={Jia, Yanwei and Zhou, Xun Yu},
  journal={Journal of Machine Learning Research},
  volume={24},
  number={161},
  pages={1--61},
  year={2023}
}

@article{jia2022policyb,
  title={Policy gradient and actor-critic learning in continuous time and space: Theory and algorithms},
  author={Jia, Yanwei and Zhou, Xun Yu},
  journal={Journal of Machine Learning Research},
  volume={23},
  number={275},
  pages={1--50},
  year={2022}
}

@article{Bo2022,
Author={Bo, Lijun and Capponi, Agostino and Zhou, Chao},
Title={Power forward performance in semimartingale markets with stochastic integrated factors}, Journal={Mathematics of Operations Research}, 
volume={48},
number={1},
pages={288-312},
year = {2023},
}

@article{vzitkovic2009dual,
  title={A dual characterization of self-generation and exponential forward performances},
  author={{\v{Z}}itkovi{\'c}, Gordan},
  journal={The Annals of Applied Probability},
  volume={19},
  number={6},
  pages={2176--2210},
  year={2009},
  publisher={Institute of Mathematical Statistics}
}

@book{ripley2009stochastic,
  title={Stochastic simulation},
  author={Ripley, Brian D},
  year={2009},
  publisher={John Wiley \& Sons}
}

@article{el2018consistent,
  title={Consistent utility of investment and consumption: a forward/backward {SPDE} viewpoint},
  author={El Karoui, Nicole and Hillairet, Caroline and Mrad, Mohamed},
  journal={Stochastics},
  volume={90},
  number={6},
  pages={927--954},
  year={2018},
  publisher={Taylor \& Francis}
}

@article{bianchi2024human,
  title={Human-robot interactions in investment decisions},
  author={Bianchi, Milo and Bri{\`e}re, Marie},
  journal={Management Science},
  year={2024},
  publisher={INFORMS}
}

@article{liang2023predictableRelative,
  title={Predictable relative forward performance processes: Multi-agent and mean field games for portfolio management},
  author={Liang, Gechun and Strub, Moris S and Wang, Yuwei},
  journal={arXiv preprint arXiv:2311.04841},
  year={2023}
}

@article{shkolnikov2016asymptotic,
  title={Asymptotic analysis of forward performance processes in incomplete markets and their ill-posed {HJB} equations},
  author={Shkolnikov, Mykhaylo and Sircar, Ronnie and Zariphopoulou, Thaleia},
  journal={SIAM Journal on Financial Mathematics},
  volume={7},
  number={1},
  pages={588--618},
  year={2016},
  publisher={SIAM}
}

@article{waldon2024forward,
  title={Forward robust portfolio selection: The binomial case},
  author={Waldon, Harrison},
  journal={Probability, Uncertainty and Quantitative Risk},
  volume={9},
  number={1},
  pages={107--122},
  year={2024},
  publisher={Probability, Uncertainty and Quantitative Risk}
}

@article{Angoshtari2024,
title = {Rank-dependent predictable forward performance processes},
author = {Bahman Angoshtari and Shida Duan},
journal = {Probability, Uncertainty and Quantitative Risk},
volume = {9},
number = {2},
pages = {181-218},
year = {2024},
publisher={Probability, Uncertainty and Quantitative Risk}
}

@article{angoshtari2022predictable,
  title={Predictable Forward Performance Processes in Complete Markets},
  author={Angoshtari, Bahman},
  journal={Probability, Uncertainty and Quantitative Risk},
  year={2023},
  volume={8},
  number={2},
  pages={141--176},
}

@article{strub2021evolution,
  title={Evolution of the {A}rrow--{P}ratt measure of risk-tolerance for predictable forward utility processes},
  author={Strub, Moris S and Zhou, Xun Yu},
  journal={Finance and Stochastics},
  volume={25},
  number={2},
  pages={331--358},
  year={2021},
  publisher={Springer}
}

@article{nadtochiy2017optimal,
  title={Optimal investment for all time horizons and {M}artin boundary of space-time diffusions},
  author={Nadtochiy, Sergey and Tehranchi, Michael},
  journal={Mathematical Finance},
  volume={27},
  number={2},
  pages={438--470},
  year={2017},
  publisher={Wiley Online Library}
}

@article{chong2019ergodic,
  title={An ergodic {BSDE} approach to forward entropic risk measures: representation and large-maturity behavior},
  author={Chong, Wing Fung and Hu, Ying and Liang, Gechun and Zariphopoulou, Thaleia},
  journal={Finance and Stochastics},
  volume={23},
  number={1},
  pages={239--273},
  year={2019},
  publisher={Springer}
}

@article{wu2024reinforcement,
  title={Reinforcement learning for continuous-time mean-variance portfolio selection in a regime-switching market},
  author={Wu, Bo and Li, Lingfei},
  journal={Journal of Economic Dynamics and Control},
  volume={158},
  pages={104787},
  year={2024},
  publisher={Elsevier}
}

@article{gao2022state,
  title={State-dependent temperature control for Langevin diffusions},
  author={Gao, Xuefeng and Xu, Zuo Quan and Zhou, Xun Yu},
  journal={SIAM Journal on Control and Optimization},
  volume={60},
  number={3},
  pages={1250--1268},
  year={2022},
  publisher={SIAM}
}

@article{guo2022entropy,
  title={Entropy regularization for mean field games with learning},
  author={Guo, Xin and Xu, Renyuan and Zariphopoulou, Thaleia},
  journal={Mathematics of Operations research},
  volume={47},
  number={4},
  pages={3239--3260},
  year={2022},
  publisher={INFORMS}
}

@article{wang2020reinforcement,
  title={Reinforcement learning in continuous time and space: A stochastic control approach},
  author={Wang, Haoran and Zariphopoulou, Thaleia and Zhou, Xun Yu},
  journal={Journal of Machine Learning Research},
  volume={21},
  number={198},
  pages={1--34},
  year={2020}
}

@article{dai2023learning,
  title={Learning equilibrium mean-variance strategy},
  author={Dai, Min and Dong, Yuchao and Jia, Yanwei},
  journal={Mathematical Finance},
  volume={33},
  number={4},
  pages={1166--1212},
  year={2023},
  publisher={Wiley Online Library}
}

@article{wang2020continuous,
  title={Continuous-time mean--variance portfolio selection: A reinforcement learning framework},
  author={Wang, Haoran and Zhou, Xun Yu},
  journal={Mathematical Finance},
  volume={30},
  number={4},
  pages={1273--1308},
  year={2020},
  publisher={Wiley Online Library}
}

@incollection{d2021robo,
  title={Robo-advising},
  author={D’Acunto, Francesco and Rossi, Alberto G},
  booktitle={The Palgrave handbook of technological finance},
  pages={725--749},
  year={2021},
  publisher={Springer}
}

@article{choulli2017explicit,
  title={Explicit description of HARA forward utilities and their optimal portfolios},
  author={Choulli, Tahir and Ma, Junfeng},
  journal={Theory of Probability \& Its Applications},
  volume={61},
  number={1},
  pages={57--93},
  year={2017},
  publisher={SIAM}
}

@article{nadtochiy2014class,
  title={A class of homothetic forward investment performance processes with non-zero volatility},
  author={Nadtochiy, Sergey and Zariphopoulou, Thaleia},
  journal={Inspired by Finance: The Musiela Festschrift},
  pages={475--504},
  year={2014},
  publisher={Springer}
}

@article{liang2017representation,
  title={Representation of homothetic forward performance processes in stochastic factor models via ergodic and infinite horizon {BSDE}},
  author={Liang, Gechun and Zariphopoulou, Thaleia},
  journal={SIAM Journal on Financial Mathematics},
  volume={8},
  number={1},
  pages={344--372},
  year={2017},
  publisher={SIAM}
}

@article{Avanesyan2020,
Author={Avanesyan, Levon and Shkolnikov, Mykhaylo and Sircar, Ronnie},
Title={Construction of a class of forward performance processes in stochastic factor models, and an extension of Widder theorem}, Journal={Finance and Stochastics}, 
VOLUME = {24},
      YEAR = {2020},
    NUMBER = {4},
     PAGES = {981–1011}
}

@article {el2014,
        TITLE = {Affine long term yield curves: an application of the {R}amsey
              rule with progressive utility},
        AUTHOR = {El Karoui, Nicole and Hillairet, Caroline and Mrad, Mohamed},
        JOURNAL = {International Journal of Financial Engineering},
    VOLUME = {1},
      YEAR = {2014},
    NUMBER = {1},
     PAGES = {1450003, 24}
      }

@article{mohamed2013exact,
  title={An Exact Connection between Two Solvable {SDEs} and a Nonlinear Utility Stochastic {PDE}},
  author={El Karoui, Nicole and Mrad, Mohamed},
  journal={SIAM Journal on Financial Mathematics},
  volume={4},
  number={1},
  pages={697},
  year={2013},
  publisher={Society for Industrial and Applied Mathematics}
}

@book{robert1999monte,
  title={Monte Carlo statistical methods},
  author={Robert, Christian P and Casella, George and Casella, George},
  volume={2},
  year={1999},
  publisher={Springer}
}

@article{devroye2006nonuniform,
  title={Nonuniform random variate generation},
  author={Devroye, Luc},
  journal={Handbooks in Operations Research and Management Science},
  volume={13},
  pages={83--121},
  year={2006},
  publisher={Elsevier}
}

@article{von1951variate,
  author  = {John von Neumann},
  title   = {Various Techniques Used in Connection with Random Digits},
  journal = {National Bureau of Standards Applied Mathematics Series},
  volume  = {12},
  pages   = {36--38},
  year    = {1951}
}

@article{liang2023predictable,
  title={Predictable forward performance processes: Infrequent evaluation and applications to human-machine interactions},
  author={Liang, Gechun and Strub, Moris S and Wang, Yuwei},
  journal={Mathematical Finance},
  volume={33},
  number={4},
  pages={1248--1286},
  year={2023},
  publisher={Wiley Online Library}
}

@incollection{musiela2016forward,
  title={Forward exponential indifference valuation in an incomplete binomial model},
  author={Musiela, Marek and Sokolova, Ekaterina and Zariphopoulou, Thaleia},
  booktitle={Advanced Modelling in Mathematical Finance},
  pages={277--302},
  year={2016},
  publisher={Springer}
}

@article{alsabah2021robo,
  title={Robo-advising: Learning investors' risk preferences via portfolio choices},
  author={Alsabah, Humoud and Capponi, Agostino and Ruiz Lacedelli, Octavio and Stern, Matt},
  journal={Journal of Financial Econometrics},
  volume={19},
  number={2},
  pages={369--392},
  year={2021},
  publisher={Oxford University Press}
}

@inproceedings{wang2021robo,
  title={Robo-advising: Enhancing investment with inverse optimization and deep reinforcement learning},
  author={Wang, Haoran and Yu, Shi},
  booktitle={2021 20th IEEE International Conference on Machine Learning and Applications (ICMLA)},
  pages={365--372},
  year={2021},
  organization={IEEE}
}

@article{capponi2022personalized,
  title={Personalized robo-advising: Enhancing investment through client interaction},
  author={Capponi, Agostino and Olafsson, Sveinn and Zariphopoulou, Thaleia},
  journal={Management Science},
  volume={68},
  number={4},
  pages={2485--2512},
  year={2022},
  publisher={INFORMS}
}

\end{document}